\documentclass[submission,copyright,creativecommons]{eptcs}
\providecommand{\event}{LSFA 2012} 
\usepackage{breakurl}             
\usepackage[english]{babel}
\providecommand{\urlalt}[2]{\href{#1}{#2}}
 \providecommand{\doi}[1]{doi:\urlalt{http://dx.doi.org/#1}{#1}}
\usepackage{graphicx}
\usepackage[all,knot]{xy}
\xyoption{arc} 

\usepackage{amsmath,amssymb,amsthm}
\usepackage{pslatex} 

\newtheorem{Thrm}{Theorem} [section]   
\newtheorem{Lem}{Lemma} [section]   
\newtheorem{Prop}{Proposition} [section] 
\newtheorem{Cor}{Corollary} [section] 
\newtheorem{Rem}{Remark} [section] 
\newtheorem{Exmpl}{Example} [section]

\newcommand{\sm}{\small} 
\newcommand{\Ds}{\displaystyle} 
\newcommand{\Ss}{\scriptstyle} 
\newcommand{\hs}{\hspace} 
\newcommand{\vs}{\vspace} 
\newcommand{\tvs}{2} 
\newcommand{\svs}{3} 
\newcommand{\nvs}{5} 
\newcommand{\mvs}{8} 
\newcommand{\lvs}{10} 
\newcommand{\Lvs}{12} 
\newcommand{\LVS}{18} 

\newcommand{\fn}{\footnote} 

\newcommand{\ul}[1]{\underline{{#1}}} 
\newcommand{\ol}[1]{\overline{{#1}}} 
\newcommand{\strl}[2]{\stackrel{#1}{#2}} 
\newcommand{\strld}[2]{\stackrel{{\Ds #1}}{#2}} 

\newcommand{\Lm}{\left} 
\newcommand{\Rm}{\right} 
\newcommand{\Ba}{\begin{array}} 
\newcommand{\Ea}{\end{array}} 
\newcommand{\Bt}{\begin{tabular}} 
\newcommand{\Et}{\end{tabular}} 
\newcommand{\Bq}{\begin{quote}} 
\newcommand{\Eq}{\end{quote}} 

\newcommand{\Bi}{\begin{itemize}} 
\newcommand{\Ei}{\end{itemize}} 
\newcommand{\Be}{\begin{enumerate}} 
\newcommand{\Ee}{\end{enumerate}} 
\newcommand{\Bd}{\begin{description}} 
\newcommand{\Ed}{\end{description}} 

\newcommand{\Lag}{\langle} 
\newcommand{\Rag}{\rangle} 
\newcommand{\Bop}{(} 
\newcommand{\Eop}{)} 
\newcommand{\BDr}{\Lag} 
\newcommand{\EDr}{\Rag} 
\newcommand{\BSl}{\Lag} 
\newcommand{\ESl}{\Rag} 
\newcommand{\drs}{,} 
\newcommand{\dps}{:} 
\newcommand{\ud}[1]{\underline{#1}} 

\newcommand{\BMd}{\Lag} 
\newcommand{\EMd}{\Rag} 

\newcommand{\seq}{=}  
\newcommand{\deq}{:\!=}  
 
\newcommand{\rar}{\rightarrow}  
\newcommand{\lngrar}{\longrightarrow}  
\newcommand{\Rar}{\Rightarrow}  
\newcommand{\Lar}{\Leftarrow}  

\newcommand{\mpt}{\mapsto}  

\newcommand{\incl}{\subseteq}  

\newcommand{\impl}{\rar}  

\newcommand{\Setof}[2]{ \{ {#1} \, : \, {#2} \} } 

\newcommand{\StC}[1]{ \ul{#1}  } 

\newcommand{\Lst}[1]{ \vec{#1}  } 

\newcommand{\sz}[1]{ | {#1} | } 
\newcommand{\lng}[1]{ lg({#1})  } 

\newcommand{\NL}{\lambda}  

\newcommand{\Nat}{{\rm I\!N}} 

\newcommand{\ES}{\emptyset} 

\newcommand{\cnc}{;} 

\newcommand{\lbl}[1]{{\mathrm #1}} 
\newcommand{\snt}[1]{\mathsf{#1}} 
\newcommand{\sem}[1]{{\mathfrak{#1}}} 
\newcommand{\Rl}[1]{{\tt #1}} 
\newcommand{\Nf}[1]{{\rm {#1}}} 
\newcommand{\LNf}[1]{{\mathrm #1}} 
\newcommand{\Snf}[1]{{\tt #1}} 

\newcommand{\cmp}{\cdot} 

\newcommand{\eqs}{\doteq} 

\newcommand{\Var}{\mathsf{V \! r}} 

\newcommand{\vx}{{\tt x}} 
\newcommand{\vy}{{\tt y}} 
\newcommand{\vz}{{\tt z}} 

\newcommand{\vu}{\Nf{u}} 
\newcommand{\vv}{\Nf{v}} 
\newcommand{\vw}{\Nf{w}} 

\newcommand{\Lu}{\LNf{u}} 
\newcommand{\Lv}{\LNf{v}} 
\newcommand{\Lw}{\LNf{w}} 

\newcommand{\Su}{\Snf{u}} 
\newcommand{\Sv}{\Snf{v}} 
\newcommand{\Sw}{\Snf{w}} 

\newcommand{\PS}{\snt{P \! r}} 
\newcommand{\ps}{\lbl{p}} 
\newcommand{\qs}{\lbl{q}} 
\newcommand{\rs}{\lbl{r}} 
\newcommand{\pss}{\lbl{s}} 
\newcommand{\ts}{\lbl{t}} 
\newcommand{\as}{\lbl{a}} 
\newcommand{\bs}{\lbl{b}} 

\newcommand{\rla}{n} 
\newcommand{\ara}{m} 
\newcommand{\art}{k} 

\newcommand{\RE}{\lbl{E}} 
\newcommand{\RF}{\lbl{F}} 

\newcommand{\cmpl}[1]{\ol{#1}} 

\newcommand{\Ff}{\varphi}  
\newcommand{\Fp}{\psi}  
\newcommand{\FP}{\Psi}  
\newcommand{\Ft}{\theta}  
\newcommand{\Fx}{\chi}  

\newcommand{\Snt}{\tau}  
 
\newcommand{\IN}{\snt{N \! m}} 

\newcommand{\arc}[2]{{{#1} \! \slash {#2} }} 
\newcommand{\frarc}[2]{{\Ds \frac{#1}{#2} }} 

\newcommand{\NS}{\snt{N}} 

\newcommand{\AS}{\snt{A}} 

\newcommand{\FA}[1]{\snt{a}[#1]} 

\newcommand{\dg}{\Sigma} 
\newcommand{\GrG}{\snt{G}} 
\newcommand{\GrH}{\snt{H}} 
\newcommand{\SlP}{\snt{P}} 
\newcommand{\SlS}{\snt{S}} 
\newcommand{\SlT}{\snt{T}} 

\newcommand{\EG}{ \{ \hs{3pt} \} } 
\newcommand{\DLS}[1]{\hat{#1}} 
\newcommand{\DLSS}{ \DLS{ \snt{s} } } 
\newcommand{\DLST}{ \DLS{ \snt{t} } } 

\newcommand{\SN}[1]{\Snf{NF}[{#1}]} 

\newcommand{\SkA}[1]{\snt{Sk}[#1]} 
\newcommand{\SlE}[1]{\snt{Sl}(#1)} 
\newcommand{\SlN}[2]{\snt{Sl}_{#2}[#1]} 
\newcommand{\NDS}[2]{\snt{DS}[{#1}\, \angle \, {#2}]} 

\newcommand{\ada}{+} 
\newcommand{\rpm}{\slash} 
\newcommand{\Brn}{[} 
\newcommand{\Ern}{]} 

\newcommand{\AlSL}[1]{\top_{#1}} 
\newcommand{\AlSS}[2]{{}^{#2} \! \top_{#1}} 

\newcommand{\ggl}[3]{{{#1}^{#2}}{#3}} 

\newcommand{\GGL}[5]{ {#1}[ #2: #3, \dots, #4:#5] } 
\newcommand{\GbL}[5]{ {#1}[ #2: #3, #4:#5] } 

\newcommand{\RZS}{{\cal Z}} 
\newcommand{\ZS}{\RZS_0} 
\newcommand{\EZS}{\RZS_*} 
\newcommand{\NZS}{\RZS_\infty} 

\newcommand{\mor}{\dasharrow} 
 
\newcommand{\cf}[1]{\nu_{#1}}  

\newcommand{\Fm}[1]{\mu_{#1}}  
\newcommand{\FM}[2]{\mu \frac{#1}{#2}}  

\newcommand{\Mor}[2]{\mathsf{Mor}[ {#1},{#2} ]} 
\newcommand{\Hom}{\rar} 
\newcommand{\iso}{\cong} 

\newcommand{\ha}{\mu'} 
\newcommand{\hb}{\mu''} 
\newcommand{\he}{\eta} 
\newcommand{\hsg}{\nu'} 
\newcommand{\hta}{\nu''} 

\newcommand{\rtcs}{ \star} 

\newcommand{\rtc}[1]{ {#1}^\rtcs} 

\newcommand{\red}{\rhd} 
\newcommand{\cnvu}[1]{\strl{(#1)}{\red}} 

\newcommand{\ex}{\lhd} 

\newcommand{\bsf}[1]{{#1}^{\tt b}} 
\newcommand{\slc}[1]{\snt{SC}(#1)} 
\newcommand{\scN}[2]{\snt{SC}_{#2}[#1]} 
\newcommand{\bsfE}[1]{\lbl{B_e}({#1})} 
\newcommand{\bsfS}[1]{\snt{B_s}[{#1}]} 
\newcommand{\bsfG}[1]{\snt{B_g}[{#1}]} 

\newcommand{\Der}{\vdash} 

\newcommand{\rl}{\Rl{r\!l}} 

\newcommand{\AtRl}{\alpha} 

\newcommand{\CndRl}{\impl} 

\newcommand{\QtRl}{Q^\ast} 

\newcommand{\EQRl}{\exists^\ast} 
\newcommand{\UQRl}{\forall^\ast} 

\newcommand{\StrSlRl}{\strl{ \SlT}{\rar}} 
\newcommand{\StrGrRl}{\strl{ \cup }{\lngrar}} 
\newcommand{\DrStrGGRl}{\strl{ \GrH }{\lngrar}} 
\newcommand{\StrFlCUn}{\cmpl{\cup}} 
\newcommand{\StrFlLu}{\uparrow} 
\newcommand{\DrStrCGRl}{\strl{ \cmpl{\GrH}}{\lngrar}} 
\newcommand{\DrStrEqarc}{\strl{ \eqs }{\rar}} 

\newcommand{\dc}{{}^{\cmpl{\cmpl{\hs{4pt}}}}} 

\newcommand{\gM}{\sem{M}} 

\newcommand{\g}{{\tt g}}  
\newcommand{\h}{{\tt h}}  

\newcommand{\asgsat}[1]{{[\![}#1{]\!]}} 

\newcommand{\rel}[1]{{{\mathsf[}}#1{{\mathsf ]}}}  
\newcommand{\arcsat}[1]{\Vdash_{#1}} 
\newcommand{\cnq}{\models} 

\newcommand{\betM}[1]{{[\![}#1{]\!]_{\gM}}} 
\newcommand{\Sk}[1]{{\tt Sk}(#1)}  
  
\newcommand{\eq}{\equiv} 

\newcommand{\SnEx}[1]{\varepsilon({#1})}  
\newcommand{\FrAr}[1]{\alpha({#1})}  
\newcommand{\FrDr}[1]{\delta({#1})}  
\newcommand{\SnSl}[1]{\sigma({#1})}  
\newcommand{\SnGr}[1]{\gamma({#1})}  
\newcommand{\FStSk}[1]{\Psi({#1})}  

\newcommand{\rkd}[1]{{\tt rk}(#1)}  
\newcommand{\rka}[1]{{\mathsf{rk}}(#1)}  
\newcommand{\rkl}[1]{{\mathsf{rk}}(#1)}  
\newcommand{\rks}[1]{{\tt rk}(#1)}  

\newcommand{\inp}{{\Ss \rightarrow}} 
\newcommand{\dsz}{2}  

\title{A Graph Calculus for Predicate Logic\thanks{Research partly sponsored by the Brazilian  agencies CNPq and FAPERJ.}}
\author{Paulo A.~S.  Veloso
\institute{COPPE-UFRJ \\ Systems and Computer Engin. Program\\
UFRJ: Federal University of Rio de Janeiro\\
RJ, Brazil}
\email{pasveloso@gmail.com}
\and
\qquad\qquad Sheila R. M. Veloso 
\institute{FEN-UERJ \\ 
Systems and Computer Engin. Dept., Fac. of Engineering \\ 
  UERJ: State University of Rio de Janeiro\\
    RJ , Brazil}
\email{
 \quad\qquad sheila.murgel.bridge@gmail.com}
}
\def\titlerunning{A Graph Calculus for Predicate Logic}
\def\authorrunning{Veloso \& Veloso}

\begin{document}

\maketitle 


\begin{abstract} 
\noindent We introduce a refutation graph calculus for classical first-order predicate logic, 
which is an extension of previous ones for binary relations. 
One reduces logical consequence  to establishing that a constructed graph has empty extension, 
 i.~e. it represents    $\bot$. 
 Our calculus establishes that a graph has empty extension by converting  it to a normal form, 
which is expanded to other graphs  until we can recognize conflicting situations 
(equivalent to a formula and its negation). 
\end{abstract}

  

\section{Introduction}   \label{sec:Intr} 
We present a  refutation graph calculus for classical first-order predicate logic. 
This approach is based on reducing logical consequence to 
showing that a constructed graph has empty extension,   representing the logical constant $\bot$.  
Our sound and complete calculus establishes when a  graph has empty extension. 

	For instance, given formulas $\Fp$, $\Ft$ and $\Ff$,  to establish that  
	$\Ff$ follows from $\{ \Fp , \Ft \}$, we construct a graph $\GrG$  corresponding to 
$\{ \Fp , \Ft  \} \cup \{ \neg \Ff  \}$ and show that $\GrG$ has empty extension. 
Now,  our calculus establishes that a graph has empty extension by converting  it to a normal form, 
which is expanded to other graphs  until we can recognize conflicting situations 
(equivalent to a formula and its negation). 

	Formulas are often 
	written down on a single line~\cite{CL_96}.     
Graph calculi rely on  two-dimensional representations  providing better visualization~\cite{CL_95}.\fn{The structure of 
$(x + y ) \cmp (z \div w )$ is more apparent  in the 
notation 
$\Lm ( \Ba{c} x \\ + \\ y \Ea  \Rm ) \cmp  \Lm ( \Ba{c} z \\ \div \\ w \Ea  \Rm )$ (see also~\cite{Bar_10}).} 
In the realm of binary relations, a simple calculus (with linear derivations)~\cite{CL_95,CL_96}  
was extended for handling complement: direct   calculi~\cite{FVVV_08,FVVV_10} and 
refutation   calculi~\cite{VV_11L}.  
Our new calculus is a further extension, 
inheriting much of the earlier terminology (such as `graph', `slice' and `arc'), 
together with some ideas from Peirce's diagrams for relations~\cite{Sw_??,Dau_06}.
The present calculus involves two new aspects:  
extension to arbitrary predicates (which affects the representation) and 
allowing formulas within the graphs.

	The structure of this paper is as follows. 
Section~\ref{sec:Motv} motivates the underlying ideas with some illustrative examples. Section~\ref{sec:GrLng} introduces our   graph  language: syntax, semantics and some constructions.  
In Section~\ref{sec:GrClc}  we introduce our graph calculus: its rules and goal. 
 Section~\ref{sec:Concl} presents some concluding remarks, including comparison with related works. 

\section{Motivation} \label{sec:Motv} 
We begin by motivating our ideas with some illustrative examples. 

	We know that consequence can be reduced to unsatisfiability. 
We will indicate how one can represent formulas graphically  and then establish consequence  by graphical means. 

	First, 
we indicate how we can represent (some) formulas graphically 
(see~\ref{subsec:Sntxsem}  for more details). 

	We represent an atomic formula by arrows  to predicate symbols coming from its arguments. 
So, we represent the formulas 
$\ps(  \vu)$ and $\rs(\vu, \vv)$, respectively, as follows: 
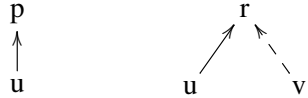
\vs{- \nvs pt}
\[  \Ba{ccc} \xymatrix@R16pt@C10pt{ & \ps \ar@{<-}[d]   &  \\ &  \vu & } 
& \hs{\LVS pt} & 
\xymatrix@R18pt@C10pt{ & \rs   & \\  \vu \ar@{->}[ur] & & \vv  \ar@{-->}[ul]  }
\Ea \]  
The former illustrates a $1$-ary arc. 
The latter is an example of a $2$-ary arc, which will be satisfied by the choices of values  $a$ and $b$ for  
 $\vu$ and $\vv$, respectively, with the pair  $\Bop a , b \Eop$ in the  $2$-ary relation interpreting $\rs$. 

	We  obtain a representation for a conjunction by joining those of its formulas. 
So, we represent the formula $\ps(  \vu) \land \rs(\vu, \vv)$ as follows: 
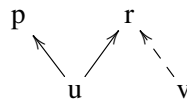
\vs{- \mvs pt} 
\[   \xymatrix@R16pt@C10pt{ \ps \ar@{<-}[dr] & & \rs   & \\ & \vu \ar@{->}[ur] & & \vv  \ar@{-->}[ul]  } \] 
 This  set of $2$ arcs is an example of a draft, 
 which also  represents the set  $\{\ps(  \vu) , \rs(\vu, \vv) \}$. 
 This draft $\snt{D}$	will be satisfied exactly by the assignments satisfying both 
 $\ps(  \vu)$ and $\rs(\vu, \vv)$. 
 
	To represent an existential quantification, we hide the node corresponding to the 
quantified variable, leaving only the rest visible.  
For instance, from formula $\ps( \vu) \land \rs(\vu, \vv)$, 
we  obtain $\exists \vx \, ( \ps( \vx) \land \rs(\vx, \vv) )$. 
We can use the representation of the former to represent the latter: we place   
the above  draft $\snt{D}$  within a box and mark $\vv$ as visible, which we represent as follows: 
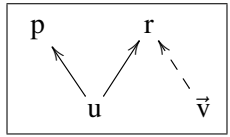
\vs{- \nvs pt}  
\[   \fbox{
\xymatrix@R18pt@C10pt{ \ps \ar@{<-}[dr] & & \rs   & \\ & \vu \ar@{->}[ur] & & \vec{\vv}  \ar@{-->}[ul]  } 
}  \] 
This is an example of a $1$-ary slice. The interpretation of this slice $\snt{S}$ is the $1$-ary relation consisting of the values $b$ such that, for some $a$, the assignment  $\vu \mpt a$, $\vv \mpt b$ satisfies the 
underlying draft $\snt{D}$.  

	Now, we can represent  formula $\neg \exists \vx \, ( \ps(  \vx) \land \rs(\vx, \vv) )$ 
by complementing this slice $\SlS$. 
As ${}^{\cmpl{\hs{\nvs pt}}}$ stands for complement, we represent  
 $\neg \exists \vx \, ( \ps(  \vx) \land \rs(\vx, \vv) )$ as follows: 
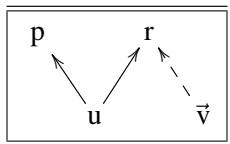
\vs{- \nvs pt}  
\[  \cmpl{ \fbox{
\xymatrix@R18pt@C10pt{ \ps \ar@{<-}[dr] & & \rs   & \\ & \vu \ar@{->}[ur] & & \vec{\vv}  \ar@{-->}[ul]  }} 
}  \] 
 
	A draft  consists of finite  sets of names and of arcs (giving constraints on the names).  
A slice consists of a draft and a list of distinguished names, which we indicate  by 
special marks, such as `$\inp$'. 
 
 	Next, we illustrate how one can establish consequence by graphical means.  
The idea is reducing unsatisfiability of a (finite) set of formulas to that of its corresponding draft. 

	We begin with an example that is basically propositional.  
Then, we examine other examples with equality $\eqs$ and existential quantifiers 
(see~\ref{subsec:Cnstr} and~\ref{subsec:Red} for more details). 

	
\begin{Exmpl}  \label{Exmpl:cnj}
Consider $\ps(\vu) \land \qs(\vu) \cnq \ps(\vu)$. 
As mentioned, we reduce it to $\{ \ps(\vu) \land \qs(\vu) , \neg \ps(\vu) \} \cnq \bot$. 
We can represent the formulas by (sets of) arcs as follows: 
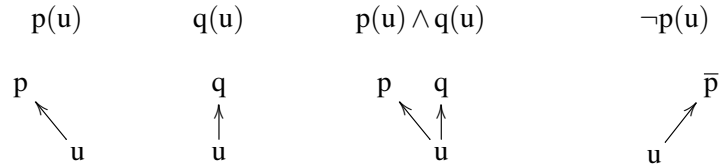
\vs{- \nvs pt}
\[ \Ba{cccccc} 
\ps(\vu) \hs{\lvs pt} & \qs(\vu) & \hs{\lvs pt} &  
\ps(\vu) \land \qs(\vu) \hs{\lvs pt} & \hs{\Lvs pt}  &  \hs{\lvs pt} \neg \ps(\vu)   \vs{\lvs pt} \\  
	\xymatrix@R14pt@C10pt{ \ps \ar@{<-}[dr] &   &  \\ & \vu & } & 
	\xymatrix@R14pt@C10pt{ & \qs \ar@{<-}[d]   &  \\ &  \vu & } 
& \hs{\lvs pt} &  
	\xymatrix@R14pt@C10pt{ \ps \ar@{<-}[dr] & \qs  \ar@{<-}[d]  & \\ & \vu & } 
& \hs{\Lvs pt}  & 
	\xymatrix@R14pt@C10pt{ &   & \cmpl{\ps}  \ar@{<-}[dl]  \\ & \vu & } 
\Ea \] 
We can obtain a representation for the set  $\{ \ps(\vu) \land \qs(\vu) , \neg \ps(\vu) \}$ by joining those of its formulas: 
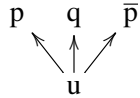
\vs{- \nvs pt} 
\[ \xymatrix@R14pt@C10pt{ \ps \ar@{<-}[dr] & \qs  \ar@{<-}[d]  & \cmpl{\ps}  \ar@{<-}[dl]  \\ & \vu & } \]
Within this draft for  $\{ \ps(\vu) \land \qs(\vu) , \neg \ps(\vu) \}$, 
we find the conflicting  situation (as ${}^{\cmpl{\hs{\nvs pt}}}$ stands for complement): 
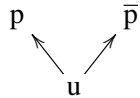
\vs{- \mvs pt}
\[ \xymatrix@R14pt@C10pt{\ps \ar@{<-}[dr] &  & \cmpl{\ps}  \ar@{<-}[dl]  \\ 
	& \vu & }  \]
Thus,  the representation of  $\{ \ps(\vu) \land \qs(\vu) , \neg \ps(\vu) \}$ 
is unsatisfiable. 
\end{Exmpl}   

\begin{Exmpl}    \label{Exmpl:notcnq}
We know that  $\ps(\vu) \not \cnq \ps(\vv)$, i.~e. 
$\{ \ps(\vu)  , \neg \ps(\vv) \} \not \cnq \bot$.  
The  corresponding draft     is: 
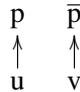
\vs{- \nvs pt} 
\[ \xymatrix@R12pt@C10pt{ \ps \ar@{<-}[d] &   \cmpl{\ps}  \ar@{<-}[d]  \\ 
	\vu & \vv  }  \] 
Here, we do not find conflicting arcs.\fn{Indeed, we have:  
\vs{- \mvs pt} 
$\Ba{ccc}
	\Ba{ccc}  
	\xymatrix@R9pt@C8pt{   \ps \ar@{<-}[d] \\  \vu }  & \Ba{c}  \\ \\  \mbox{ but not } \Ea &  
	\xymatrix@R9pt@C8pt{  \cmpl{\ps} \ar@{<-}[d] \\  \vu }  \Ea 
& \hs{\Lvs pt} \mbox{ and }  \hs{\Lvs pt}  & 
	\Ba{ccc}  
	\xymatrix@R9pt@C8pt{  \cmpl{\ps} \ar@{<-}[d] \\  \vv }  & \Ba{c}  \\ \\  \mbox{ but not } \Ea &  
	\xymatrix@R9pt@C8pt{  \ps \ar@{<-}[d] \\  \vv }  \Ea 
\Ea$\!\!\!.}  
In fact, we can read from the representation a model $\gM = \BMd M , \ps^\gM \EMd$,  
with $M \deq \{ \vu, \vv \}$ and  $\ps^\gM \deq \{ \vu \}$, where  one can satisfy  $\ps(\vu)$ and $\neg \ps(\vv)$. 
\end{Exmpl}   

	
\begin{Exmpl}    \label{Exmpl:eq} 
We reduce $\ps(\vv) \land \vv \eqs \vu \cnq \ps(\vu)$ 
to  the unsatisfiability of the set $\{ \ps(\vv) \land  \vv \eqs \vu , \neg \ps(\vu) \}$. 
We have the graphical representations as sets of arcs as follows: 
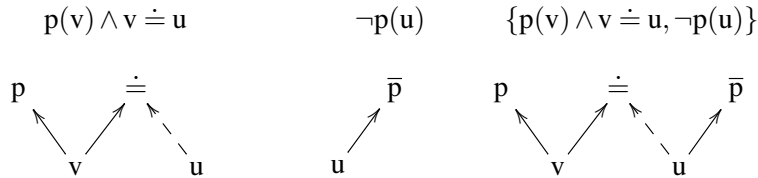
\vs{- \nvs pt} 
\[ \Ba{ccccc} \ps(\vv) \land  \vv \eqs \vu \hs{\lvs pt} & & \hs{\LVS pt} \neg \ps(\vu) & &
	\hs{\lvs pt}  \{ \ps(\vv) \land \vv \eqs \vu , \neg \ps(\vu) \} \vs{\lvs pt} \\ 
\xymatrix@R16pt@C10pt{ 
	\ps \ar@{<-}[dr]  &  & \ar@{<-}[dl]  \eqs \ar@{<--}[dr]  &  &  \\ 
	 & \vv &  & \vu &  } &  &
	\xymatrix@R16pt@C10pt{  
	  & \ar@{<-}[dl]   \cmpl{\ps} \\ 
	 \vu &  }  & & 
	\xymatrix@R16pt@C10pt{  
	\ps \ar@{<-}[dr]  &  & \ar@{<-}[dl]  \eqs \ar@{<--}[dr]  &  & \ar@{<-}[dl]   \cmpl{\ps} \\ 
	 & \vv &  & \vu &  }
\Ea \] 
Now, we can simplify the representation of $\{ \ps(\vv) \land  \vv \eqs \vu , \neg \ps(\vu) \}$, 
by renaming $\vv$ to $\vu$:  
\vs{- \lvs pt} 
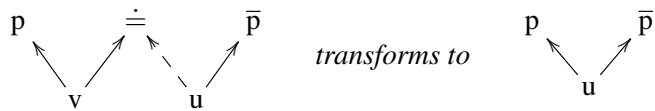
\[ \Ba{ccc} 
	\xymatrix@R16pt@C10pt{  
	\ps \ar@{<-}[dr]  &  & \ar@{<-}[dl]  \eqs \ar@{<--}[dr]  &  & \ar@{<-}[dl]   \cmpl{\ps} \\ 
	 & \vv &  & \vu &  }
&  \Ba{c}  \\ \\  \mbox{ transforms to } \Ea  & 
	\xymatrix@R12pt@C10pt{  \ps \ar@{<-}[dr] &  & \cmpl{\ps}  \ar@{<-}[dl]  \\ 
	& \vu & }
\Ea \] 
This final representation is not satisfiable (cf. Example~\ref{Exmpl:cnj}). 
\end{Exmpl}   

	
\begin{Exmpl}    \label{Exmpl:exst} 
We reduce $\rs(\vu, \vv) \cnq \exists \vz \, \rs( \vu , \vz)$  
to  $\{ \rs(\vu, \vv) , \neg \exists \vz   \rs( \vu , \vz) \} \cnq \bot$. 
As before, we can represent formula $\rs( \vu , \vv)$ by the single-arc draft: 
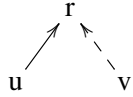
\vs{- \svs pt} 
\[ \xymatrix@R17pt@C10pt{ & \rs   & \\  \vu \ar@{->}[ur] & & \vv  \ar@{-->}[ul]  } \]
Also, we can represent  $\neg \exists \vz   \rs( \vu , \vz)$ by the following $1$-ary arc:  
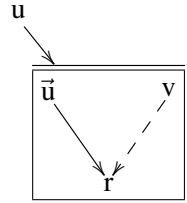
\vs{- \nvs pt} 
\[ \xy 
(0,0)++*{\vu}="u";  (0.2,-0.5)++*{}="u+"; (12,-16)*{ \cmpl{\fbox{\xy 
(0,0)+*{\vec{\vu}}="u"; (0.2,-0.5)++*{}="u+";  (8,-12.7)*{ \rs  }="r" ; 
 (16,0)*{\vv}="w" ;  (15.5,-1.5)*{}="w-" ;  (7.4,-11.4)*{ }="r-" ; (8.6,-11.4)*{ }="r+" ; 
 "u+";"r-" ** \dir{-} ?>* \dir{>};  "w-";"r+" ** \dir{--} ?>* \dir{>}; 
\endxy } } }="v" ; "u+";"v"**\dir{-} ?>* \dir{>}
\endxy \] 
Thus, we can represent  $\{ \rs(\vu, \vv) , \neg \exists \vz \rs( \vu , \vz) \}$  by the draft:   
 \vs{- \svs pt}   
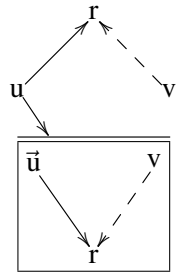
\[ \xy (10,10.4)*{\rs}="r" ;  (9.2,9.2)*{}="r-" ; (10.8,9.2)*{}="r+" ; 
(20,0)*{\vv}="v" ; (19.0,0.8)*{}="v-" ;  ;  "v-";"r+" ** \dir{--} ?>* \dir{>}; 
(0,0)+*{\vu \, }="u"; (10,-15)*{ \cmpl{  \fbox{\xy 
(0,0)+*{\vec{\vu}}="u"; (0.2,-0.5)++*{}="u+";  (8,-12.7)*{ \rs  }="r" ; 
 (16,0)*{\vv}="w" ;  (15.5,-1.5)*{}="w-" ;  (7.4,-11.4)*{ }="r-" ; (8.6,-11.4)*{ }="r+" ; 
 "u+";"r-" ** \dir{-} ?>* \dir{>};  "w-";"r+" ** \dir{--} ?>* \dir{>}; 
\endxy } }}="u'" **\dir{-} ?>* \dir{>}; 
"u";"r-" ** \dir{-} ?>* \dir{>}
\endxy \] 
Now, with  $\vec{\vu} \mpt \vu ,\vv \mpt \vv$,  we have a copy of the slice under complement within the draft, namely: 
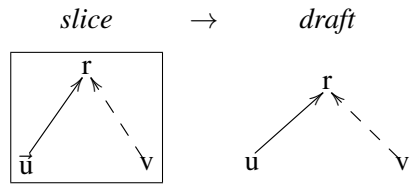
\vs{- \svs pt} 
\[ \Ba{ccc} \mbox{slice} &  \rar  & \mbox{draft} \vs{\nvs pt}  \\ 
\fbox{\xy 
(0,0)+*{\vec{\vu}}="u";   (0.2,0.7)+*{}="u+";  (8,12.8)*{\rs}="r"; 
(16,0)*{\vv}="w"; (15.5,0.7)*{}="w-"; (7.4,11.4)*{ }="r-" ; (8.6,11.4)*{ }="r+" ; 
 "u+";"r-"**\dir{-} ?>* \dir{>};   "w-";"r+"**\dir{--} ?>* \dir{>}
\endxy } & & 
\xy (0,0)+*{\vu}="u";   (0.2,0.7)+*{}="u+"; (10,10.8)*{\rs}="r";  
(20,0)*{\vv}="v" ; (19.5,0.7)*{}="v-"; (9.4,9.4)*{ }="r-" ; (10.6,9.4)*{ }="r+" ;   
  "u+";"r-"**\dir{-} ?>* \dir{>};    "v-";"r+"**\dir{--} ?>* \dir{>}  
\endxy
\Ea \] 
So, the  representation of $\{ \rs(\vu, \vv) , \neg \exists \vz \rs( \vu , \vz) \}$ is not satisfiable.  
\end{Exmpl}   
	
	
\begin{Exmpl}    \label{Exmpl:exsts2}
We reduce $\exists \vx \exists \vy [ \rs(\vu, \vx) \land \pss(\vx, \vy) ] \cnq \exists \vz \rs(\vu, \vz)$  to  
$\{ \exists \vx \exists \vy [ \rs(\vu, \vx) \land \pss(\vx, \vy) ] , \neg  \exists \vz \rs(\vu, \vz) \} \cnq \bot$. 
Proceeding as before, we can be represent 
$\{ \exists \vz  \exists \vy [ \rs(\vx, \vz) \land \pss(\vz, \vy) ] , \neg \exists \vz \rs(\vu, \vz) \}$ as: 
\vs{- \tvs pt} 
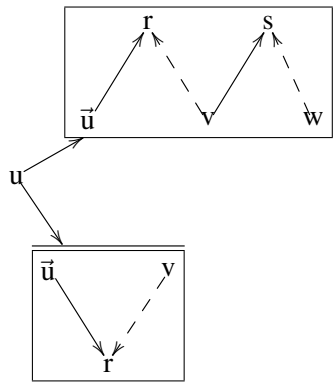
\[ \xy (24,14)*{ \fbox{ \xy 
(0,0)+*{ \vec{\vu}}="x"; (0.5,0.5)+*{}="x+";  (8,12.8)*{\rs}="r";  
 (16,0)*{\vv}="z" ;   (15.3,0.5)*{}="z-" ; (16.7,0.5)*{}="z+" ;  (24,12.8)*{\pss}="s" ; 
 (30,0)*{\vw}="y"  ;  (29.5,0.5)*{}="y+"  ;   (7.4,11.4)*{}="r-";   (8.6,11.4)*{}="r+";  
 (23.4,11.4)*{}="s-" ;   (24.6,11.4)*{}="s+" ; 
   "x+";"r-"**\dir{-} ?>* \dir{>};     "z-";"r+"**\dir{--} ?>* \dir{>}; 
    "z+";"s-"**\dir{-} ?>* \dir{>};  "y+";"s+"**\dir{--} ?>* \dir{>} 
\endxy} };  (0.8,0.8)**\dir{-} ?<* \dir{<} ; 
 (0,0)+*{\vu \, }="x";  (12,-18)*{ \cmpl{\fbox{\xy 
 (0,0)+*{ \vec{\vu}}="x"; (0.5,-0.6)+*{ }="x+"; (8,-12.8)*{\rs}="r" ;  
(16,0)*{\vv}="v" ;  (15.5,-1.0)*{}="v-" ;  (7.4,-11.4)*{}="r-" ;  (8.6,-11.4)*{}="r+" ; 
"x+";"r-"**\dir{-} ?>* \dir{>}; "v-";"r+"**\dir{--} ?>* \dir{>} 
\endxy }}}="v" ; (0.4,-0.8)**\dir{-} ?<* \dir{<} 
\endxy \] 
We can transform this representation into the following one,  which is, much as before, unsatisfiable. 
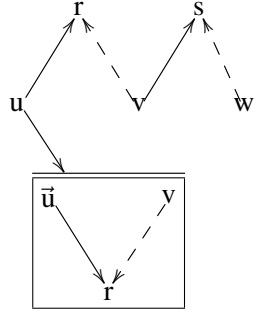
\vs{- \tvs pt} 
\[ \xy 
(0,0)+*{ \vu \, }="x"; (0.5,0.5)+*{}="x+"; (8,12.8)*{\rs}="r" ;  
 (16,0)*{\vv}="z" ;   (15.3,0.5)*{}="z-" ; (16.7,0.5)*{}="z+" ;  (24,12.8)*{\pss}="s";    
 (30,0)*{\vw}="y" ;   (29.5,0.5)*{}="y+"  ; 
  (7.4,11.4)*{}="r-" ;  (8.6,11.4)*{}="r+" ; 
  (23.4,11.4)*{}="s-";   (24.6,11.4)*{}="s+";  
 (0,0)+*{ };  (12,-18)*{ \cmpl{\fbox{\xy 
(0,0)+*{ \vec{\vu}}="x"; (0.5,-0.6)+*{ }="x+"; (8,-12.8)*{\rs}="r" ;  
(16,0)*{\vv}="v" ;  (15.5,-1.0)*{}="v-" ;  (7.4,-11.4)*{}="r-" ;  (8.6,-11.4)*{}="r+" ; 
"x+";"r-"**\dir{-} ?>* \dir{>}; "v-";"r+"**\dir{--} ?>* \dir{>} 
\endxy }}}="v"; 
  (0.8,-0.8)**\dir{-} ?<* \dir{<} ;  
"x+";"r-"**\dir{-} ?>* \dir{>};   "z-";"r+"**\dir{--} ?>* \dir{>}; 
 "z+";"s-"**\dir{-} ?>* \dir{>};  "y+";"s+"**\dir{--} ?>* \dir{>} 
\endxy \] 
\end{Exmpl}   

\section{Graph  Language} \label{sec:GrLng} 
We now introduce our  concepts: expressions, slices and graphs will give relations, 
whereas arcs, sketches and drafts will correspond to constraints.  
We will  examine syntax and semantics (in~\ref{subsec:Sntxsem})   and then 
some concepts and constructions  (in~\ref{subsec:Cnstr}).  

	We first introduce some notations.
		Given a function $f  : A \rar B$, we use $f(a)$ or $a^f$ for 
its \emph{value} at an element $a \in A$; which we extend to lists and sets. 
For a list $a \seq \Lag a_1, \dots, a_\art \Rag \in A^\art$, we  use 
$f(a)$ or $a^f$ for the  \emph{list of values}  $\Lag {a_1}^f, \dots, {a_\art}^f\Rag \in B^\art$; 
for a set  $\NS$, we  use 
$f(\NS)$ or $\NS^f$ for the  \emph{set of values}  $\Setof{a^f}{ a \in \NS}$. 
Given a list $a = \Lag a_1, \dots, a_\art \Rag \in A^\art$, we   employ  
$\StC{a}$ 
for its \emph{set of components}. 
The \emph{null list} is $\NL \deq \Lag \hs{\svs pt} \Rag$. 
We sometimes write a list $\Lag a_1, \dots, a_\art \Rag$ simply as $a_1\, \dots \, a_\art$. 

 	We will use names (or parameters) for marking free places and 
variables for marking bound places, 
as usual in Proof Theory~\cite{Tak_75}.  
To quantify a formula $\Ff$ we replace a name $\vu$ by a new variable (not appearing in $\Ff$) 
obtaining $\exists \vx \, \Ff \Brn \vu \rpm \vx \Ern$ and  $\forall \vx \, \Ff \Brn \vu \rpm \vx \Ern$. 
Also, given  lists $\Lu$, of $\rla$ distinct names, and $\vx$, of $\rla$ distinct variables not occurring in $\Ff$, 
we have the formulas $\exists^\rla \vx \, \Ff \Brn \Lu \rpm \vx \Ern$ and  $\forall^\rla \vx \, \Ff \Brn \Lu \rpm \vx \Ern$.     
	
	We will consider first-order predicate languages (without function symbols, except the constant $\bot$), each one characterized by 
 pairwise disjoint sets as follows:    
\Bd 
\item[{\sm ($\IN$)}] 
an infinite linearly ordered \emph{set  of names} $\IN$; 
\item[{\sm ($\Var$)}] 
a  denumerably infinite \emph{set of variables} $\Var$;   
\item[{\sm ($\PS$)}] 
(possibly empty, but  pairwise disjoint) sets $\PS_\rla$ of  $\rla$-ary \emph{predicate symbols}, for $\rla \in \Nat$. 
\Ed 

	Given $\ara \in \Nat_+$, we use  $\vu_{\ara}$ for the \emph{$\ara$th name}.  
Given $\rla \in \Nat$, we use  $\Lu^{\rla} \deq \Lag  \vu_{1}, \dots  \vu_{\rla} \Rag$ for the 
\emph{list of the first $\rla$ names} (with $\Lu^{0} = \NL$).
Also, given a set $\Sv \incl \IN$ of names, we use $\Lst{\Sv}$ for the list of the names in $\vv$  in the ordering of 
$\IN$. 
For a formula $\Ff$, we use $\SN{\Ff}$ for the \emph{set of names occurring in} $\Ff$.

  
 \subsection{Syntax and semantics} \label{subsec:Sntxsem}
 We now introduce  the syntax and semantics of our concepts.
We first examine  the syntax of our concepts.  

	The objects of our graph language are defined  by mutual recursion  as follows.  
\Bd 
\item[{\sm ($\RE$)}] An $\rla$-ary \emph{expression}  is   
an $\rla$-ary predicate symbol, 
a formula with $\rla$  names,  
an $\rla$-ary slice or graph (see below), or 
 $\cmpl{\RE}$,  where $\RE$ is an $\rla$-ary expression. 
 For instance,  $\bot$ is a $0$-ary expression, 
 $\eqs$ and $\cmpl{\eqs}$ are $2$-ary expressions, whereas 
$\ps(\vu)$  and $\cmpl{\ps}$ (for $\ps \in \PS_1$) are $1$-ary expressions.   

 \item[{\sm ($\snt{a}$)}] An $\ara$-ary \emph{arc} $\snt{a}$ over set $\NS \incl \IN$ of names  
 is a pair 
 $\arc{\RE}{\Lv}$ (also noted  $\frarc{\RE}{\Lv}$), 
where $\RE$ is  an $\ara$-ary expression and   $\Lv \in \NS^\ara$. 
Examples 
are $\arc{\bot}{\NL}$,   
$\arc{\eqs}{ \vu \, \vv }$, 
$\arc{\cmpl{\ps}}{ \vu }$, $\arc{\qs(\vu)}{ \vv }$ (for $\ps, \qs \in \PS_1$) and 
$\arc{\pss}{ \vu \, \vw }$ (for $\pss \in \PS_2$).

\item[{\sm ($\dg$)}] A \emph{sketch} $\dg \seq \BDr \NS \drs \AS \EDr$   
consists of sets $\NS \incl \IN$ of \emph{names} and $\AS$ of  \emph{arcs} over $\NS$. 
 
\item[{\sm ($\snt{D}$)}] A \emph{draft} $\snt{D} \seq \BDr \NS \drs \AS \EDr$ is a sketch with 
finite sets $\NS$ of names and $\AS$ of arcs. 
An example of draft is $\snt{D}' \seq  \BDr \{ \vu, \vu', \vv, \vw, \vw' \}  \drs 
\{ \arc{\cmpl{\ps}}{ \vu }, \arc{\qs(\vu)}{ \vv }, \arc{\eqs}{ \vw \, \vw' },
\arc{\pss}{ \vu \, \vw } \} \EDr$. 

\item[{\sm ($\snt{S}$)}]  An $\rla$-ary \emph{slice} $\SlS \seq \BSl \ul{\SlS} \dps \DLSS  \ESl$   
consists of 
its \emph{underlying draft}   $\ud{\snt{S}} \deq \BDr \NS \drs \AS \EDr$  and a 
 \emph{distinguished list}  $\DLSS$, with $\DLSS \in \NS^\rla$. 
 For instance,  $\SlS \seq \BSl\{ \vu, \vu', \vv, \vw, \vw' \}  \drs 
\{ \arc{\cmpl{\ps}}{ \vu }, \arc{\qs(\vu)}{ \vv }, \arc{\eqs}{ \vw \, \vw' },
\arc{\pss}{ \vu \, \vw } \} \dps  \vu \, \vv \, \vv  \ESl$ is a $3$-ary slice with 
underlying draft $\ud{\snt{S}} \seq \snt{D}'$ (as above) and distinguished list $\DLSS \seq \Lag \vu, \vv, \vv \Rag$.
 
\item[{\sm ($\snt{G}$)}]  An $\rla$-ary \emph{graph}   is a finite set of $\rla$-ary slices. 

\Ed  
In particular, the empty graph $\{ \hs{3pt} \}$ has no slice. 
Example~\ref{Exmpl:Exprl} (in~\ref{subsec:Drv}) will show a $2$-slice graph. 
	
	Note that expressions, arcs, slices and graphs are finite objects,  
whereas sketches are not necessarily so. Sketches will be useful for representing models and constructing co-limits. Also, some concepts and results do not depend on finiteness (see~\ref{subsec:Cnstr}), 
which will be important in Section~\ref{sec:GrClc}.
We wish to represent these finite objects graphically by drawings   
(cf. the examples in Section~\ref{sec:Motv}).    
For this purpose, we employ two sorts of nodes: 
name nodes (labeled by names) and  expression nodes (labeled by expressions).  
Some representations aiming at  precision and readability are as follows. 

	We represent an $\ara$-ary arc $\arc{\RE}{\vv}$, with $\Lv \seq \Lag \vv_1, \dots,  \vv_\ara \Rag$, by  
$\ara$ arrows connecting each node labeled by $\vv_i$ to the node labeled by $\RE$.  
For instance, we can draw a  $3$-ary arc $\arc{\ts}{\Lag \vu,  \vv, \vw \Rag}$ as 
{\sm $\Ba{ccccc} 
&  &  \ts  & &  \\ 
&   \nearrow &  \uparrow &  \nwarrow  \\ 
 \vu & & \vv & & \vw 
\Ea $}.   
To clarify  (as in  $\frarc{\ts}{ \vu \, \vv \, \vu }$), we may use distinct kinds of lines or label  them  by numbers.\fn{We often employ  full, dashed, dotted and wavy  lines, respectively, for the 1st,  2nd,   3rd and 4th arguments of expressions.}  
A more compact version uses is $\vv_1 \, \strld{\rs}{\rar} \,  \vv_2$ for the $2$-arc $\arc{\rs}{\vv_1  \vv_2}$,  
representing $1$-ary,  $3$-ary and $4$-ary arcs, respectively, as:  
$
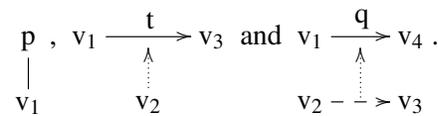
\xymatrix@R12pt@C8pt{ \ps \\ \vv_1   \ar@{-}[u] }$, 
$\xymatrix@R12pt@C8pt{
\vv_1 \ar[rr]^{{\Ds \,\ts }} & & \vv_3\\
&\vv_2  \ar@{.>}[u] & }$  and 
$\xymatrix@R12pt@C8pt{
\vv_1 \ar[rr]^{{\Ds \,\qs }} & & \vv_4\\
\vv_2  \ar@{-->}[rr] & \ar@{.>}[u] & \vv_3 }$.

	We can indicate the components of a distinguished list by marking their nodes, say  with numbers, 
e. g. $\Lag \vu, \vv, \vu \Rag$ by $ \vu^{1, 3} \vv^2$. 
Also, it may be  convenient  (for easier visualization) to enclose a slice $\SlS$ within a full box, \fbox{$\SlS$}, and 
a graph $\GrG$ within a dashed box, \dashbox{\dsz}(15, 15)[]{$\GrG$}. 
For instance, Example~\ref{Exmpl:exsts2} (in Section~\ref{sec:Motv}) shows a $0$-ary slice 
$\SlS \seq \BSl \{  \vu , \vv, \vw \}  \drs \AS \dps  \NL \ESl$, with 
$\AS \seq  \{ \arc{\rs}{  \vu \, \vv } , \arc{\pss}{  \vv \, \vw } , \arc{\cmpl{\SlT}}{ \vu } \}$, 
where $\SlT$ is the $1$-ary slice 
$\BSl \{  \vu , \vv \}  \drs \{ \arc{\rs}{  \vu \, \vv } \} \dps  \vu  \ESl$. 


 	Given a list $\Lw$ of names, the \emph{arcless $\vw$ slice} is the  slice 
$\AlSL{\Lw} \deq \BSl  \StC{\Lw} \drs \ES \dps   \Lw \ESl$.  
The \emph{arcless $\ara$-ary slice} is the  slice $ \AlSL{\ara} \deq  \AlSL{\Lu^{\ara}} $ 
($\Lu^{\ara}$ is the list of the first $\ara$ names) and the \emph{$\ara$-node arcless  draft} is 
$ \ud{\AlSL{\ara}} \seq \BDr \StC{ \vu^{\ara}} \drs \ES \EDr$.
 The \emph{arc of  formula} $\Ff$, with set $\vv$ of names, is $\FA{\Ff} \deq \arc{\Ff}{\Lst{\vv}}$. 
 (The arc of  a sentence $\Snt$ is  $0$-ary: $\FA{\Snt} \seq \arc{\Snt}{\NL}$, which we represent as 
the expression node $\Snt$.) 
The \emph{sketch of the set of arcs}  $\AS$ is the sketch 
$\SkA{\AS} \deq \BDr \NS \drs \AS \EDr$,  where $\NS$ consists of the names occurring in the arcs of 
$\AS$: $\NS \deq \bigcup \Setof{\StC{\Lw} \incl \IN}{\arc{\RE}{\Lw} \in \AS}$. 
For instance, $\SkA{\{ \arc{\pss}{\vu \, \vv} , \arc{\ps(\vv)}{\vw} \}} \seq 
	\BDr \{ \vu, \vv, \vw \} \drs \{ \arc{\pss}{\vu \, \vv} , \arc{\ps(\vv)}{\vw} \} \EDr$.

	We may wish to add an arc $\snt{a} \seq \arc{\RE}{\Lv}$ to a sketch, a slice or a  graph.  
For a  sketch  $\dg \seq \BDr \NS \drs \AS \EDr$,  we set 
$\dg \ada  \snt{a} \deq \BDr \NS  \cup \StC{\vv} \drs \AS \cup \{ \snt{a} \}  \EDr$; 
for a slice $\snt{S} \seq \BSl \ul{\snt{S}} \dps \DLSS \ESl$,  we set  
	$\SlS \ada  \snt{a} \deq  \BSl \ul{\snt{S}} \ada  \snt{a} \dps \DLSS \ESl$; 
for a graph  $\snt{G}$,  we set  $\snt{G} \ada  \snt{a} \deq \Setof{\SlS \ada  \snt{a}}{\SlS \in \snt{G}}$. 
The \emph{difference slice} of a finite set of arcs $\AS$ with respect to 
an arc $\snt{a} \seq \arc{\RE}{\Lv}$ is the $0$-ary slice 
$\NDS{\AS}{\snt{a}} \deq \BSl \SkA{\AS}  \ada \arc{\cmpl{\RE}}{\Lv} \dps \NL \ESl$. 
For instance, Example~\ref{Exmpl:eq} (in Section~\ref{sec:Motv}) represents the set  
$\{ \ps(\vv) \land  \vv \eqs \vu , \neg \ps(\vu) \}$ by the difference slice 
$\NDS{\{ \arc{\ps}{ \vv },  \arc{\eqs}{ \vv \, \vu } \} }{\arc{\ps}{ \vu }}$. We will give some intuition for using   $0$-ary slices~in \ref{subsec:Cnstr}.

	We now examine  the semantics  of our  concepts and related ideas. 

	A \emph{model}  $\gM$ has as its  universe a set  
$M \neq \ES$ and 
realizes each $\rla$-ary predicate symbol $\ps \in \PS_\rla$ as an $\rla$-ary relation  
${\ps}^{\gM} \subseteq M^\rla$ (with  $\eqs^\gM \deq \Setof{\Lag a , a \Rag \in M^2}{a \in M}$).  
An \emph{$M$-assignment} for  set $\NS \incl \IN$ of names is a function  $\g : \NS \rar M$.  
A formula $\Ff$ with  set $\Sv$ of $\rla$ names  \emph{defines} 
the $\rla$-ary relation  $\Ff^\gM  \subseteq M^\rla$ consisting of the values of its 
ordered names   for the  assignments satisfying $\Ff$: 
$\Ff^\gM \deq \Setof{{\Lst{\Lv}}^\h \in M^\rla}{ \gM \models \Ff \, \asgsat{\h}}$. 
For instance, for $2$-ary predicate symbol $\rs$, 
$\rs( \vu_1 , \vu_2)^\gM \seq \rs^{\gM}$,  
$\rs( \vu_2 , \vu_1)^\gM \seq \Setof{\Lag b , a \Rag \in M^2}{\Lag a , b \Rag \in {\rs}^{\gM}}$ and 
$\rs( \vu_1 , \vu_1)^\gM \seq \Setof{\Lag  a \Rag \in M^1}{\Lag a , a  \Rag \in {\rs}^{\gM}}$. 
Also, $\bot^\gM \deq \ES$.

	We  now  introduce the meanings of the concepts, again  by mutual recursion. 	
\Bd 
\item[{\sm ($\RE$)}] We define the \emph{relation} of an expression    as follows. 
For a predicate symbol $\ps$ we have its relation: $\rel{\ps}_{\gM} \deq {\ps}^{\gM}$;  
for formula $\Ff$ we have its defined relation: $\rel{\Ff}_{\gM} \deq {\Ff}^{\gM}$;    
for a  slice $\SlS$ or  graph $\snt{G}$, 
we use the extensions:    
$\rel{\snt{S}}_{\gM} \deq \betM{\SlS}$ and $\rel{\snt{G}}_{\gM} \deq  \betM{\GrG}$ (see below);  
for $\cmpl{\RE}$,  where $\RE$ is an $\rla$-ary expression,  
	we use the complement: 
	$\rel{\cmpl{\RE}}_{\gM} \deq M^\rla \setminus \rel{\RE}_{\gM}$.  

\item[{\sm ($\snt{a}$)}] An M-assignment $\g : \NS \rar M$  \emph{satisfies} an  
$\ara$-ary arc $\arc{\RE}{\Lv}$ over $\NS$ in $\gM$   
(noted $\g \arcsat{\gM} \arc{\RE}{\Lv}$) iff 
$\Lv \in \NS^\ara$ and  
$\Lv^\g \in \rel{\RE}_{\gM}$.  For instance, 
$\g \arcsat{\gM} \arc{\eqs}{ \vu \, \vv }$ iff $\vu^\g = \vv^\g$ and 
$\g \arcsat{\gM} \arc{\ps}{ \vw }$ iff $\vw^\g \in \rel{\ps}_{\gM} = {\ps}^{\gM}$. 

\item[{\sm ($\dg$)}]  An assignment $\g$  \emph{satisfies} a sketch  
$\dg \seq \BDr \NS \drs \AS \EDr$   in $\gM$ 
(noted $\g : \dg \rar \gM$) iff $\g$ satisfies every arc $\snt{a} \in \AS$. 

\item[{\sm ($\snt{S}$)}] The \emph{extension} of a slice is the relation consisting of values of its distinguished list    for the  assignments satisfying its underlying draft;  for  
an $\rla$-ary  slice 
$\snt{S} \seq \BSl \ul{\SlS} \dps \DLSS \ESl$, 
$\betM{\snt{S}} \deq \Setof{ \DLSS^\g  \in M^\rla}{\g : \ud{\snt{S}}  \rar \gM}$. 

\item[{\sm ($\snt{G}$)}]  The \emph{extension} of a  graph is the union of those of 
its slices: 
    $\betM{\snt{G}} \deq \bigcup_{\snt{S} \in \snt{G}}\betM{\snt{S}}$.

\Ed 
Clearly,   $\g \arcsat{\gM} \arc{\cmpl{\RE}}{\Lv}$ iff  $\g \not \arcsat{\gM} \arc{\RE}{\Lv}$. 
Also, the arcless $\ara$-ary slice  $\AlSL{\ara} $ has extension $\betM{ \AlSL{\ara} } \seq M^\ara$. 

	An expression $\RE$ is \emph{null} iff $\rel{\RE}_{\gM} \seq \ES$ in every model $\gM$. 
For instance, the empty graph $\EG$ is null.
	
	Given a sketch  $\dg \seq \BDr \NS \drs  \AS \EDr$  and an arc $\snt{a} \seq \arc{\RE}{\Lv}$, we say that 
 $\snt{a}$ \emph{is a consequence of} $\dg$ 
 (noted $\dg \cnq \snt{a}$)  iff, for every model $\gM$ and 
  $M$-assignment $\g  :  \NS \cup \StC{\Lv} \rar M$, 
$\g$ satisfies $\snt{a}$ whenever $\g$ satisfies $\dg$. 
Call expressions $\RE$ and $\RF$ \emph{equivalent}  (noted $\RE \eq \RF$) iff, 
for every model $\gM$,  $\rel{\RE}_{\gM} \seq \rel{\RF}_{\gM}$. 
A slice $\SlS$ and the singleton graph  $\{ \SlS \}$ are equivalent 
(so they may be identified). 
 
	We can reduce  consequence  to the difference slice: 
an arc $\snt{a}$ is a consequence of a draft $\snt{D}$ iff  
 the difference slice $\NDS{\AS}{\snt{a}}$ is null. 
So, we can also reduce logical consequence  to a difference slice.\fn{Recall that 
$\FP \cnq \Ft$  iff,  for every model $\gM$ and  assignment $\h$, 
$\h$ satisfies $\Ft$ whenever $\h$ satisfies every  $\Fp \in\FP$.}  
	
\begin{Prop}  \label{Prop:LgCnDS}
 Given  a finite set $\FP$ of formulas  and a formula $\Ft$:
$\FP \cnq \Ft$ iff       the difference slice 
$\NDS{\Setof{\FA{\Fp}}{\Fp \in \FP}}{\FA{\Ft}}$ is null.  
\end{Prop}   
\begin{proof} 
By  the preceding remark, since $\g \arcsat{\gM} \FA{\Ff}$ iff  $\gM \models \Ff \, \asgsat{\g}$. 
\end{proof}   

	Section~\ref{sec:GrClc} will present a calculus for establishing that an expression is null.
	
\subsection{Concepts and constructions} \label{subsec:Cnstr} 
We now examine some concepts and constructions. 

 	We first introduce  morphisms for comparing sketches.   
	
	Consider sketches  $\dg' \seq \BDr \NS' \drs  \AS' \EDr$  and $\dg'' \seq \BDr \NS'' \drs  \AS'' \EDr$. 
A function $\he: \NS'' \rar \NS'$ is a \emph{morphism} from $\dg''$ to $\dg'$   
(noted $\he: \dg'' \mor \dg'$) iff it preserves arcs:   
for every arc $\arc{\RE}{\Lv} \in \AS''$, we have $\arc{\RE}{ \Lv^\he } \in \AS'$. 
We use $\Mor{\dg''}{\dg'}$ for the \emph{set of  morphisms}  from $\dg''$ to $\dg'$.  

\begin{Exmpl}  \label{Exmpl:mor} 
Given $\ps \in \Pr_1$ and $\qs, \rs, \pss, \ts, \as, \bs \in \Pr_2$, consider the  drafts 
$\snt{D}' \seq \BDr \NS'  \drs \AS' \EDr$ and $\snt{D}'' \seq \BDr \NS''  \drs \AS'' \EDr$, 
with sets of nodes 
$\NS' \seq \{ \vu , \vv, \vv' , \vw , \vw' \}$ and 
$\NS'' \seq \{ \vu_1 ,  \vu_2 ,  \vu_3 ,  \vv, \vv_1 , \vv_2 , \vw , \vw_1, \vw_2 , \vw' \}$,  
and sets of arcs 
\[ \AS' \seq \{ \arc{\qs}{\vv \, \vw} , \arc{\ps}{\vw'} , \arc{\rs}{\vv \, \vw'}  , \arc{\pss}{\vv \, \vu} , 
	 \arc{\ts}{\vu \, \vw}  , \arc{\as}{\vu \, \vv'} , \arc{\bs}{\vv' \, \vw} \} \mbox{ and }  \] 
\[ \AS'' \seq \{ \arc{\qs}{\vv_1 \, \vw_1} , \arc{\qs}{\vv_2 \, \vw_2} ,  \arc{\ps}{\vw'} ,  
	\arc{\rs}{\vv \, \vw'} , \arc{\rs}{\vv_1 \, \vw'} ,  \arc{\rs}{\vv_2 \, \vw'} , \arc{\pss}{\vv_2 \, \vu_3} , 
	 \arc{\ts}{\vu_2 \, \vw_1}  , \arc{\as}{\vu_1 \, \vv'} , \arc{\as}{\vu_3 \, \vv'} ,  
	 \arc{\bs}{\vv' \, \vw} ,  \arc{\bs}{\vv' \, \vw_1} ,  \arc{\bs}{\vv' \, \vw_2} \}. \]  
These drafts $\snt{D}'$ and $\snt{D}''$ can be represented as 
 in Figure~\ref{Fig:DrD}.  
The mapping $\vv' \mpt \vv'$;  $\vw' \mpt \vw'$; $\vv, \vv_1 ,\vv_2 \mpt \vv$; 
 $\vw, \vw_1 ,\vw_2 \mpt \vw$  and $\vu_1, \vu_2 ,\vu_3 \mpt \vu$ 
preserves arcs.\fn{For 
 instance, for arc $\arc{\ps}{\vw}$ of $\snt{D}''$, we have 
 arc $\arc{\ps}{\vw'}$ of $\snt{D}'$;  
 for arcs $\arc{\qs}{\vv_1 \, \vw_1}$ and $\arc{\qs}{\vv_2 \, \vw_2}$  of $\snt{D}''$, 
 we have arc $\arc{\qs}{\vv \, \vw}$ of $\snt{D}'$.}  
So, we  have a morphism   $\he: \snt{D}'' \mor \snt{D}'$.  
We also have formulas $\FrDr{\snt{D}'}$ and $\FrDr{\snt{D}''}$ such that 
$\g : \snt{D}' \rar \gM$ iff $\gM \models \FrDr{\snt{D}'}  \, \asgsat{\g}$ and 
$\g : \snt{D}'' \rar \gM$ iff $\gM \models  \FrDr{\snt{D}''}  \, \asgsat{\g}$.\fn{Take 
$\FrDr{\snt{D}'}$ as 
$\qs(\vv , \vw)  \land \ps(\vw') \land  \rs(\vv , \vw')  \land 
	 \pss(\vv , \vu) \land  \ts(\vu , \vw) \land \as(\vu , \vv') \land \bs(\vv' , \vw) $ and 
$\FrDr{\snt{D}''}$ as the conjunction of $\qs(\vv_1 , \vw_1)$, $\qs(\vv_2 , \vw_2)$, $\ps(\vw')$,    
	$\rs(\vv , \vw')$, $\rs(\vv_1 , \vw')$, $\rs(\vv_2 , \vw')$, $\pss(\vv_2 , \vu_3)$,    
	 $\ts(\vu_2 , \vw_1)$, $\as(\vu_1 , \vv')$,  $\as(\vu_3 , \vv')$,  
	$\bs(\vv' , \vw)$,  $\bs(\vv' , \vw_1)$ and  $\bs(\vv' , \vw_2)$.} 
\end{Exmpl} 

\begin{figure}[htb] 
\vs{- \svs pt}
\[ \Ba{ccc} \snt{D}' & \hs{\LVS pt}  \hs{\LVS pt}   &  \snt{D}'' \vs{\mvs pt} \\
\xymatrix@R26pt@C20pt{ 
\ps \ar@{-}[r] &\vw'& \vv\ar[l]_{{\Ds \rs}} \ar[r]^{{\Ds \qs}}\ar[d]_{{\Ds \pss}}&\vw\\
&&\vu\ar[ur]^{{\Ds \ts}}\ar[r]^{{\Ds \as}}
&\vv'\ar[u]_{{\Ds \bs}}
}
& &
\xymatrix@R20pt@C22pt{ 
&\vv_1\ar[rr]^{{\Ds \qs}}\ar[d]_{{\Ds \rs}}&&\vw_1&&\ar[ll]_{{\Ds \ts}}\vu_2\\
\ps \ar@{-}[r]&\vw'&\ar[l]_{{\Ds \rs}}\vv_2\ar[r]^{{\Ds \qs}}\ar[dr]_{{\Ds \pss}}&\vw_2&\vv'\ar[l]_{{\Ds \bs}}\ar[ul]_{{\Ds \bs}}\ar[r]^{{\Ds \bs}}&\vw\\
&\vv\ar[u]^{{\Ds \rs}}&&\vu_3\ar[ur]_{{\Ds \as}}&&\vu_1\ar[ul]_{{\Ds \as}}\\
}
\Ea
\]
\caption{Drafts $\snt{D}'$ and $\snt{D}''$ (Example~\ref{Exmpl:mor})} \label{Fig:DrD}
\end{figure}

	 A  morphism transfers satisfying assignments by composition. 
	 
\begin{Lem} \label{Lem:Hmtrnsfr} 
Given a   morphism     $\he: \dg'' \mor \dg'$,  
for every assignment $\g : \NS_{\dg'} \rar M$ satisfying  $\dg'$, 
the composite  $\g \cmp \he : \NS_{\dg''} \rar M$ is an assignment satisfying  $\dg''$. 
\end{Lem}
 \begin{proof} 
 For every arc  $\arc{\RE}{\Lv} \in \AS_{\dg''}$,  we have 
  $\arc{\RE}{ {\Lv^\he }} \in \AS_{\dg'}$, thus      $\Lv^{\g \cmp \he} \in \rel{\RE}_{\gM}$, whence  
  $\g \cmp \he \arcsat{\gM} \arc{\RE}{\Lv}$.     
 \end{proof}  
 
  	We now use morphisms to introduce  zero    sketches, slices and graphs. 
	
	A  sketch  $\dg \seq \BDr \NS \drs  \AS \EDr$ is   \emph{zero} iff   there exist 
a slice $\SlT \seq \BSl \ud{\SlT} \dps \DLST \ESl$ and a morphism 
$\he : \ud{\SlT} \mor \dg$ such that 
$\arc{\cmpl{\SlT}}{\DLST^\he}$ is an arc in $\AS$. 
A slice $\SlS$ is  \emph{zero} iff its underlying draft $\ud{\SlS}$ is a zero sketch. 
A graph  is  \emph{zero} iff  all its slices are  zero slices. 
The sets of zero drafts,  zero slices and zero graphs are all decidable, since, 
for drafts $\snt{D}'$ and $\snt{D}''$, 
the set $\Mor{\snt{D}''}{\snt{D}'}$ is finite.


\begin{Exmpl}    \label{Exmpl:ZrDrft} 
Consider the following draft $\snt{D} $ and $2$-ary slice $\SlT$:  
\vs{- \nvs pt} 
\[ \Ba{ccc} 
\snt{D} \seq \BDr \{ \vu', \vv', \vw' \} \drs 
	\{ \arc{\rs}{\vu' \, \vv'},  \arc{\cmpl{\SlT}}{\vu' \, \vw'}, \arc{\pss}{\vv' \, \vw'}  \} \EDr & \hs{\LVS pt} & 
\SlT \seq \BSl \{ \vu , \vv , \vw  \}  \drs \{ \arc{\rs}{\vu \, \vv},   \arc{\pss}{\vv \, \vw} \} \dps \vu \, \vw \ESl 
\vs{\svs pt} \\ 
\xymatrix@R18pt@C18pt{ 
\vu'\ar[r]^{{\Ds \rs}}\ar[dr]_{{\Ds \cmpl{\SlT} }}&\vv'\ar[d]^{{\Ds \pss }}\\
&\vw'
}
& & 
\xymatrix@R18pt@C18pt{ 
\vu^1\ar[r]^{{\Ds \rs}} & \vv\ar[d]^{{\Ds \pss}}  \\
&   \vw^2
} 
\Ea \]
The mapping $\vu \mpt \vu'$, $\vv \mpt \vv'$, $\vw \mpt \vw'$ gives a morphism $\he : \ud{\SlT} \mor \snt{D}$ , 
with $\DLST^\he = \Lag \vu^\he , \vw^\he \Rag = \Lag \vu', \vw' \Rag$. 
Thus, draft $\snt{D}$ is zero. So, slices 
$\BSl \snt{D} \dps \NL  \ESl$, $\BSl \snt{D} \dps \vu'  \ESl$, $\BSl \snt{D} \dps \vv' \, \vw'  \ESl$ and 
$\BSl \snt{D} \dps \vu'  \ESl$, $\BSl \snt{D} \dps \vu'  \vv' \, \vw'  \ESl$ are  zero slices.\fn{The extension of slice $\SlT$ can be described by the formula 
$\exists \vy \, ( \rs(\vu , \vy) \land   \pss(\vy , \vw) )$.} 
\end{Exmpl}   


\begin{Lem} \label{Lem:ZrSktch} 
No assignment can satisfy a  zero sketch. 
\end{Lem}
 \begin{proof} 
 By Lemma~\ref{Lem:Hmtrnsfr}, $\g : \dg \rar \gM$ yields $\g \cmp \he :  \ud{\SlT} \rar \gM$, thus  
 	$\g \arcsat{\gM} \arc{\SlT}{\DLST^\he}$ whence $\g \not \arcsat{\gM} \arc{\cmpl{\SlT}}{\DLST^\he}$.    
 \end{proof}   
 
 \begin{Cor} \label{Lem:ZrSlGr} 
Zero   slices and zero graphs are null. 
\end{Cor}  
\begin{proof} 
 By 
 Lemma~\ref{Lem:ZrSktch}: if  $\betM{\SlS} \neq \emptyset$,  then 
 some assignment  satisfies  $\ud{\SlS}$. 
 \end{proof} 
 
	We can now clarify the intuition behind using   $0$-ary difference slices (cf.~\ref{subsec:Sntxsem}). 
We know that a formula is satisfiable iff its existential closure is so. 
The latter 
will convert to a $0$-ary (basic) graph, by Proposition~\ref{Prop:Red} 
(in~\ref{subsec:Red}). Now, whether a slice is zero does not hinge on its distinguished list.  
 
	We now examine some categorical constructions: co-limits and pushouts~\cite{Mac_71}. 
	
 	The category of sketches and morphisms has co-limits. 
Given a  diagram of  sketches $\dg_i \seq \BDr \NS_i \drs  \AS_i \EDr$, its  \emph{co-limit} 
can be obtained as expected: 
obtain the co-limit $\NS$ of the sets of names $\NS_i$ and then 
transfer arcs, by 
the functions $\cf{i} : \NS_i \rar \NS$, i.~e. 
 $\AS \deq \bigcup_{i \in I} \, {\AS_i}^{\cf{i}}$.     
In particular,  the pushout of drafts gives a draft.  
 
 	We wish to glue a slice $\SlT$ onto  a draft or a slice  via a designated list of names. 
This involves  adding the arcs of $\SlT$ with its distinguished list identified to the  designated list of names.  


 	Gluing can be introduced as an amalgamated sum (of drafts). 
  Consider an  $\ara$-ary slice $\SlT \seq \BSl \ud{\SlT} \dps \DLST \ESl$. Given a draft 
 $\snt{D} \seq \BDr \NS \drs \AS \EDr$ and a list  
$\Lw \in \NS^\ara$ of  $\ara$ names, 
the \emph{glued draft} $\ggl{\snt{D}} {\Lw}  {\SlT}$ is the pushout of the  drafts 
$\snt{D} \, \ada \,\StC{\Lw} \deq \BDr \NS \cup \StC{\Lw}  \ \drs \AS \EDr$ and 
$\ud{\SlT}$ over the $\ara$-ary arcless draft $\ud{\AlSL{\ara} } \seq \BDr \StC{ \Lu^{\ara}} \drs \ES \EDr$  
and the natural morphisms $\ha$ and $\hb$  
($\ha: \vu_i \mpt \vw_i$ and $\hb: \vu_i \mpt  \DLST_i$),  as shown 
in Figure~\ref{Fig:PODr}. 
Note that $\hsg(\Lw) = \hta(\DLST)$.  
 \begin{figure}[htb]   
 \[   \xymatrix@R12pt@C20pt{ 
 & \snt{D} \, \ada \,\StC{\Lw}  \ar@{-->}[rd] ^{{\Ds \hsg }}  & \\ 
\BDr  \StC{ \Lu^{\ara}} \drs \ES \EDr \ar@{-->}[ru] ^{{\Ds \ha }}  \ar@{-->}[rd] _{{\Ds \hb }} 
	& &  \ggl{\snt{D}} {\Lw}  {\SlT}  \\ 
 &  \ud{\snt{T}}  \ar@{-->}[ru] _{{\Ds \hta }} &  
}  \]  
\caption{Pushout of drafts} \label{Fig:PODr}
\end{figure}  


	Given an $\rla$-ary slice $\SlS  \seq \BSl  \ud{\snt{S}} \dps \DLSS \ESl$, we obtain the  
\emph{glued slice} $\ggl{\SlS} {\Lw}  {\SlT}$ 
by transferring the distinguished list of $\SlS$ to the glued draft $\ggl{\ud{\SlS}} {\Lw}  {\SlT}$: 
 $\ggl{\SlS} {\Lw}  {\SlT} \deq \BSl  \ggl{\ud{\SlS}} {\Lw}  {\SlT} \dps \hsg(\DLSS) \ESl$.\fn{A 
 glued  draft and slice are unique up to isomorphism. 
They can be made unique by a suitable choice of names. 
 As isomorphic objects have the same behavior, 
we often consider a sketch or a slice up to isomorphism.} 
We glue a graph by gluing its slices, i.~e. 
 $\ggl{\SlS} {\Lw}  {\GrH}$ is the graph $\Setof{ \ggl{\SlS} {\Lw}  {\SlT} }{\SlT \in \GrH}$. 
We glue onto a  graph by gluing onto its slices,  i.~e. 
 $\ggl{\GrG} {\Lw}  {\GrH} \deq \bigcup_{\SlS \in \GrG} \, \ggl{\SlS} {\Lw}  {\GrH}$. 


\begin{Exmpl}  \label{Exmpl:glue}  
Consider the three slices:  
$1$-ary $\SlS \seq \BSl\{  \vu, \vu', \vv', \vv  \}  \drs \{ \arc{\rs}{ \vu \, \vu' }, 
		\arc{\pss}{ \vu' \, \vv' }, \arc{\ts}{ \vv' \, \vv } \}  \dps  \vu   \ESl$ 
as well as  $2$-ary  
$\SlT'  \seq \BSl\{  \vv, \vw \}  \drs \{ \arc{\as}{ \vw \, \vv } \}  \dps  \vw \, \vv  \ESl$ and 
$\SlT''   \seq \BSl\{  \vw \}  \drs \{ \arc{\ps}{ \vw }, \arc{\qs}{ \vw } \}  \dps  \vw \,  \vw  \ESl$.\fn{The extension of slice 
$\SlT''$ can be described by the formula 
$\ps(\vw) \land \qs(\vw) \land \vw \eqs \vw'$.}  
They are represented as follows: 
\vs{- \nvs pt}  
\[  \Ba{ccccc} 
\SlS & \hs{\LVS pt}  & \SlT' &  \hs{\LVS pt}  & \SlT''  \vs{\nvs pt} \\ 
\xymatrix@R11pt@C18pt{ 
 \vu^1\ar[r]^{{\Ds \rs}}    & \vu'\ar[r]^{{\Ds \pss}}   & \vv' \ar[r]^{{\Ds \ts}}  & \vv
 }
&&
\xymatrix@R11pt@C18pt{
 \vw^1 \ar[r]^{{\Ds \as}} &   \vv^2
}&&
\xymatrix@R11pt@C14pt{
&  \vw^{1, 2} &  \\  \ps \ar@{-}[ur]  &  & \ar@{-}[ul]  \qs}
\Ea \] 

\noindent 
We obtain $1$-ary glued slices as follows: 
%
\vs{- \nvs pt} 
\[ \Ba{ccc} 
 \ggl{\SlS} {\Lag \vu' , \vv' \Rag}  {\SlT'} \mbox{{\sm $\seq 
	 \BSl\{  \vu, \vu', \vv', \vv \}  \drs 
	  \Lm \{ \Ba{c}  \arc{\rs}{ \vu \, \vu' }, 
		\arc{\pss}{ \vu' \, \vv' }, \arc{\ts}{ \vv' \, \vv },  \\ \arc{\as}{ \vu' \, \vv' } \Ea \Rm  \}  
		 \dps  \vu   \ESl$}} 
&  \hs{\svs pt}  & \ggl{\SlS} {\Lag \vu' , \vv' \Rag}  {\SlT''} \mbox{{\sm $\seq 
	 \BSl \Lm \{ \Ba{c}   \vu, \vv, \\ \vw  \Ea \Rm  \}  \drs 
	 \Lm \{ \Ba{c}  \arc{\rs}{\vu \, \vw}, \arc{\pss}{\vw \, \vw}, \arc{\ts}{\vw \, \vv}, \\ 
	 	 \arc{\ps}{\vw},   \arc{\qs}{\vw}  \Ea \Rm  \}  
	 \dps  \vu \ESl$}} \vs{\lvs pt}  \\  

\xy
(0,0)*+{ \vu^1 }="u1";
(15,0)*+{  \vu'}="ul"; 
(30,0)*+{   \vv' }="vl";
(45,0)*+{ \vv}="v";
{\ar^{{\Ds \rs}}@/_-0pc/"u1";"ul"}; 
{\ar^{{\Ds \pss}}@/_-1pc/"ul";"vl"};
{\ar_{{\Ds \as}}@/^-1pc/"ul";"vl"}; 
{\ar^{{\Ds \ts}}@/_-0pc/"vl";"v"}
\endxy& &
 \xymatrix@R16pt@C18pt{
 \vu^1\ar[r]^{{\Ds \rs}}&\vw \ar@(ul,ur)[]^{{\Ds \pss}}\ar[r]^{{\Ds \ts}} \ar@{-}[dl] \ar@{-}[dr]&\vv\\
 \ps&&\qs
 }
\Ea \] 
\end{Exmpl} 


	Addition of a slice-arc is equivalent to gluing the slice. 
For instance, with the slices of Example~\ref{Exmpl:glue}: 
$\SlS \ada \arc{\SlT'}{\vu' \, \vv'} \, \eq \, \ggl{\SlS} {\Lag \vu', \vv' \Rag}  {\SlT'}$ and 
$\SlS \ada \arc{\SlT''}{\vu' \, \vv'} \, \eq \, \ggl{\SlS} {\Lag \vu', \vv' \Rag}  {\SlT''}$. 

\begin{Prop} \label{Prp:Arcadgl} 
Given a   slice $\SlS$ and an arc   $\arc{\SlT}{\Lw}$: 
$\SlS \ada \arc{\SlT}{\Lw} \, \eq \, \ggl{\SlS} {\Lw} {\SlT}$.
\end{Prop}  
\begin{proof} 
 By Lemma~\ref{Lem:Hmtrnsfr} and the pushout  property .
 \end{proof}   

	It is not difficult to translate our graph language to the 
underlying first-order predicate language.
It suffices to express the semantics of  the graph language  (in~\ref{subsec:Sntxsem}) by  
formulas.
	

\section{Graph Calculus} \label{sec:GrClc}  
We now introduce our graph calculus, with conversion  and expansion rules.  
We employ $\rtc{R}$ for the \emph{reflexive-transitive closure} of a binary relation $R$ on a set, 
as usual.   

	Our conversion and expansion rules will transform an expression to an equivalent one. 
Thus, one can apply such a rule in any context. 
For instance, we will have a rule converting $\bot$ to the empty graph   
$\{ \hs{3pt} \}$;  so, we can apply it to convert 
$\cmpl{\bot}$ to   $\cmpl{\{ \hs{3pt} \}}$ and 
$\SlS \ada \arc{\bot}{\NL}$ to $\SlS \ada \arc{\{ \hs{3pt} \}}{\NL}$, for any slice $\SlS$. 
Also, we can identify a singleton graph with its slice  (cf.~\ref{subsec:Sntxsem}): 
if $\SlS \, \red \, \RF$ then $\{ \SlS \} \, \red \, \RF$ and 
if $\RE \, \red \, \SlT$ then $\RE \, \red \, \{ \SlT \}$. 

\subsection{Conversion} \label{subsec:Red}  
We now introduce the basic objects and the conversion rules. 

	The \emph{basic} objects are defined (by mutual recursion) as follows. 
The basic expressions  are   
the 
predicate symbols, other than $\eqs$,  and    
 $\cmpl{\SlT}$,  where $\SlT$ is a basic slice (see below). 
An arc  $\arc{\RE}{\Lv}$ is basic iff  $\RE$ is a basic expression.   
A sketch is basic iff all its arcs are basic. 
A slice is basic iff its  underlying draft is a basic    sketch.  
 A graph is basic iff  its slices are all basic.  
For instance, the drafts 
$\snt{D}'$ and $\snt{D}''$,   of  Example~\ref{Exmpl:mor},  
and $\snt{D}$, of Example~\ref{Exmpl:ZrDrft},   
(in~\ref{subsec:Cnstr})  are basic, whereas those in Examples~\ref{Exmpl:cnj},~\ref{Exmpl:notcnq} 
and~\ref{Exmpl:eq} 
are not basic. 

	The conversion rules will transform an expression to an equivalent basic graph. 

	The formula rules will come from some equivalences between formulas and expressions  
We  now illustrate some of these equivalences. 
For a $1$-ary predicate $\ps$, formula $\ps(\vv)$ is equivalent to the $1$-ary slice 
$\BSl \{ \vv \}  \drs \{\arc{\ps}{\vv} \} \dps \vv \ESl$, thus  
$\neg \ps(\vv)$ is equivalent to the $1$-ary expression $\cmpl{\ps(\vv)}$. 
Now, consider formulas $\rs( \vu, \vv )$ and $\pss( \vv , \vw )$. 
For the conjunction $\rs( \vu, \vv ) \land \pss( \vv , \vw )$, we have a $3$-ary slice $\SlS$ equivalent to  it, 
namely  the slice 
$\SlS \seq \BSl \NS \drs \AS \dps  \vu \, \vv \, \vw  \ESl$, with sets   
$\NS \seq \{ \vu, \vv, \vw \}$ and $\AS \seq \{ \arc{\rs( \vu, \vv )}{  \vu \, \vv } , \arc{\pss( \vv, \vw )}{ \vv \, \vw } \}$. 
For the disjunction $\rs( \vu, \vv ) \lor \pss( \vv , \vw )$ we have a $3$-ary graph $\GrG$ such that  
$\rs( \vu, \vv ) \lor \pss( \vv , \vw ) \eq \GrG$, namely  the graph $\GrG$ with $2$  slices:  
$\BSl \{ \vu, \vv, \vw \} \drs \{ \arc{\rs( \vu, \vv )}{  \vu \, \vv } \} \dps  \vu \, \vv \, \vw  \ESl$ and 
$\BSl \{ \vu, \vv, \vw \} \drs \{ \arc{\pss( \vv, \vw )}{ \vv \, \vw } \} \dps  \vu \, \vv \, \vw  \ESl$. 
Also, as  the conditional formula $\rs( \vu, \vv ) \impl \pss( \vv , \vw )$ is logically equivalent to 
$\neg \rs( \vu, \vv ) \lor \pss( \vv , \vw )$, it is   equivalent to the $3$-ary graph   
$\{ \BSl \{ \vu, \vv, \vw \} \drs \{ \arc{\cmpl{\rs( \vu, \vv )}}{  \vu \, \vv } \} \dps  \vu \, \vv \, \vw  \ESl ,  
\BSl \{ \vu, \vv, \vw \} \drs \{ \arc{\pss( \vv, \vw )}{ \vv \, \vw } \} \dps  \vu \, \vv \, \vw  \ESl \}$. 
The existential formula $\exists \vy \, \ts( \vu, \vy , \vw )$ is equivalent to the  $2$-ary slice 
$\BSl \{ \vu, \vv, \vw \} \drs \{ \arc{\ts( \vu, \vv , \vw )}{  \vu \, \vv \,  \vw } \} \dps  \vu \,  \vw  \ESl$. 
Also, as  the universal formula $\forall \vy \, \ts( \vu, \vy , \vw )$ is logically equivalent to 
$\neg \exists \vy \neg \ts( \vu, \vy , \vw )$, 
it is   equivalent to the $2$-ary expression 
$\cmpl{\BSl \{ \vu, \vv, \vw \} \drs \{ \arc{\cmpl{\ts( \vu, \vv , \vw )}}{  \vu \, \vv \,  \vw } \} \dps  \vu \,  \vw  \ESl}$.

	The \emph{formula rules} are the following  $8$   
	conversion rules  eliminating formulas. 
\Bd 

\item[{\sm ($\AtRl$)}]  
 For an atomic formula $\ps(\Lw)$: 
$ \ps(\Lw) \, \red \,   \BSl \StC{\Lw} \drs \{\arc{\ps}{\Lw} \} \dps \Lw \ESl$. 
So, we replace $\vu \eqs \vv$, $\rs( \vu, \vv )$ and $\ts( \vu, \vv, \vv )$ 
by  
$\Ba{cccc} 
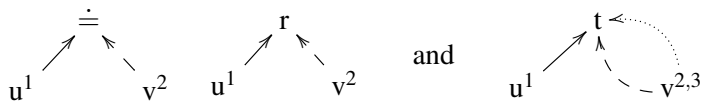
\xymatrix@R12pt@C10pt{ &  \ar@{<-}[dl] \eqs  \ar@{<--}[dr] &  \\ \vu^1 &  & \vv^2} &
\xymatrix@R12pt@C10pt{ &  \ar@{<-}[dl] \rs  \ar@{<--}[dr] &  \\ \vu^1 &  & \vv^2}  & 
 \Ba{c} \\ \\ \mbox{ and } \Ea & 
\xymatrix@R12pt@C16pt{ &  \ar@{<-}[dl] \ts  \ar@{<--}@/_1pc/[dr]  \ar@{<.}@/^1pc/[dr] &  \\ 
		\vu^1 &  & \vv^{2, 3}} 
\Ea$

\item[{\sm ($\bot$)}] 
$\bot \, \red \, \EG$,  i.~e.  we replace $0$-ary  formula $\bot$ by the empty graph. 

\item[{\sm ($\neg$)}]  $\neg \Ff \, \red \, \cmpl{\Ff}$. 
So, we replace $\neg ( \rs( \vu, \vv ) \impl \pss( \vv , \vw ) )$ by the $3$-ary expression 
$\cmpl{\rs( \vu, \vv ) \rar \pss( \vv , \vw )}$. 

\item[{\sm $\bullet$}] 
Given formulas $\Fp$ and $\Ft$, with $\Su \deq \SN{\Fp}$ and $\Sv \deq \SN{\Ft}$, 
set $\Sw \deq \Su \cup \Sv$. 

\item[{\sm ($\land$)}] 
$\Fp \land \Ft \, \red \, \BSl \Sw \drs \{  \arc{\Fp}{\Lst{\Su}} ,  \arc{\Ft}{\Lst{\Sv}}  \} \dps \Lst{\Sw} \ESl$.  
Thus, we can replace formula 
$\rs( \vu, \vv ) \land \pss( \vv , \vv )$ by the $2$-ary slice \\ 
$\BSl \{ \vu, \vv \} \drs \{ \arc{\rs( \vu, \vv )}{  \vu \, \vv } , \arc{\pss( \vv, \vv )}{ \vv } \} \dps  \vu \, \vv   \ESl$, 
which we can represent as:      
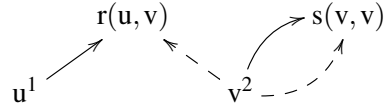
\vs{- \nvs pt}
$$\xymatrix@R13pt@C16pt{ 
	& \ar@{<-}[dl]   \rs(\vu, \vv)   \ar@{<--}[dr] & &  \ar@{<-}@/^-1pc/[dl]  \pss(\vv, \vv) \ar@{<--}@/_-1pc/[dl]   \\ 
	\vu^1 &  & \vv^2 &  }$$



\item[{\sm ($\lor$)}] 
$\Fp \lor \Ft \, \red \, \{ \BSl \Sw \drs \{  \arc{\Fp}{\Lst{\Su}}  \} \dps \Lst{\Sw} \ESl \, , \, 
\BSl \Sw \drs \{  \arc{\Ft}{\Lst{\Sv}}  \} \dps \Lst{\Sw} \ESl \}$.    
So,  we can replace  formula 
$\rs( \vu, \vv ) \lor \pss( \vv , \vv )$ by the $2$-ary graph 
$\Lm \{ \Ba{c} 
	\BSl \{ \vu, \vv \} \drs \{  \arc{\rs( \vu, \vv)}{  \vu \, \vv }  \} \dps  \vu \, \vv  \ESl \, , \\ 
	\BSl  \{ \vu, \vv \} \drs \arc{\pss( \vv, \vv)}{ \vv }  \} \dps   \vu \, \vv \  \ESl  
	\Ea \Rm  \} $.\fn{This graph  can be represented  as follows: 
\vs{- \nvs pt}    
\[ \dashbox{\dsz}(234,54)[]{$\Ba{ccc}
\fbox{\xymatrix@R13pt@C16pt{ 
&&\\
\vu^1 \ar[rr]^{{\Ds \,\rs(\vu,\vv)}}&&\vv^2
}
 }& \hs{\mvs pt}  & 
 \fbox{ \xymatrix@R10pt@C13pt{ 
	&   &   \ar@{<-}@/^-1pc/[dl]  \pss(\vv, \vv) \ar@{<--}@/_-1pc/[dl]   \\ 
	\vu^1 &   \vv^2 & } }

\Ea$
}\]}

\item[{\sm ($\CndRl$)}] 
$\Fp \impl \Ft \, \red \,\BSl \Sw \drs \{  \arc{\cmpl{\Fp}}{\Lst{\Su}} ,  \arc{\Ft}{\Lst{\Sv}}  \} \dps  \Lst{\Sw}  \ESl$. 
So,  we can replace  formula 
$\ps( \vu) \impl \rs( \vv , \vw )$ by the $3$-ary graph 
$ \{ \BSl \{ \vu, \vv , \vw \} \drs \{  \arc{\cmpl{\ps( \vu)}}{  \vu  }  \} \dps  \vu \, \vv  \, \vw  \ESl \, , 
	\BSl  \{ \vu, \vv , \vw \} \drs \arc{\rs( \vv , \vw)}{ \vv \, \vw }  \} \dps  \vu \, \vv  \, \vw   \ESl   \} $.  

\item[{\sm $\bullet$}] 
Given a formula $\Ff$  and a set $\Sv$ of names, set $\Su \deq \Sw \setminus \Sv$,  where $\Sw \deq \SN{\Ff}$. 

\item[{\sm ($\EQRl$)}] 
For formula $\EQRl \vx \, \Ff \Brn \Lv \rpm \vx \Ern$, 
$\EQRl \vx \, \Ff \Brn \Lv \rpm \vx \Ern  \, \red \,  \BSl \Sw  \drs  \{\arc{\Ff}{\Lst{\Sw}} \}   \dps  \Lst{\Su}  \ESl$.  
Thus, we can replace 
$\exists \vy \exists \vz \,  \ts( \vu, \vy , \vz )$ by  the single-arc $1$-ary slice   
$\BSl \{ \vu, \vv , \vw \}  \drs  \{\arc{\ts( \vu, \vv , \vw )}{ \vu \, \vv \, \vw } \}  \dps  \vu \ESl$, 
which we can represent as: 

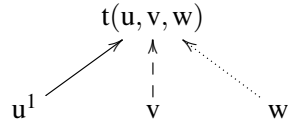
\vs{- \mvs pt} 
\[ \xymatrix@R19pt@C18pt{ 
&  \ar@{<-}[dl] \ts(\vu, \vv, \vw)  \ar@{<--}[d]  \ar@{<.}[dr] &  \\ \vu^1 & \vv & \vw} \]  

\item[{\sm ($\UQRl$)}] 
 For formula $\UQRl \vx \, \Ff \Brn \Lv \rpm \vx \Ern$, 
$\UQRl \vx \, \Ff \Brn \Lv \rpm \vx \Ern  \, \red \, 
	\cmpl{\BSl \Sw  \drs  \{\arc{\cmpl{\Ff}}{\Lst{\Sw}} \}   \dps  \Lst{\Su}  \ESl}$. 
So, can we replace 
$\forall \vy \forall \vz \,  \ts( \vu, \vy , \vz )$ by  the  $1$-ary expression    
$\cmpl{\BSl \{ \vu, \vv , \vw \}  \drs  \{\arc{\cmpl{\ts( \vu, \vv , \vw )}}{ \vu \, \vv \, \vw } \}  \dps  \vu \ESl}$, 
which we can represent as:  

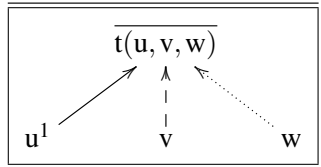
\vs{- \mvs pt} 
\[ \cmpl{\fbox{\xymatrix@R19pt@C18pt{ 
&  \ar@{<-}[dl] \cmpl{\ts(\vu, \vv, \vw)}  \ar@{<--}[d]  \ar@{<.}[dr] &  \\ \vu^1 & \vv & \vw}}} \]  

\Ed	


\begin{Exmpl}  \label{Exmpl:DtQntCnv}  
Consider a formula $\Ff$ with list of names $\Lag \vu , \vv , \vw \Rag$, noted $\Ff( \vu , \vv , \vw )$. \\ 
For the formula $\exists \vy \forall \vz \, \Ff( \vu , \vy , \vz )$, we have the conversions: 
\vs{- \mvs pt} 
\[  \exists \vy \forall \vz \, \Ff( \vu , \vy , \vz ) \, \cnvu{\EQRl} \, 
 \BSl \{ \vu , \vv \}   \drs  \{\frarc{ \forall \vz \, \Ff( \vu , \vv , \vz ) }{ \vu \, \vv } \}   \dps \vu  \ESl  \, \cnvu{\UQRl} \, 
  \BSl \{ \vu , \vv \}   \drs  
  	\{ \frarc{  \cmpl{\BSl \{  \vu , \vv , \vw \}   \drs  
		\{ \arc{\cmpl{\Ff}}{\vu \, \vv \, \vw} \}   \dps \vu \, \vv  \ESl}  }{ \vu \, \vv } \}   \dps \vu  \ESl \] 
For the formula $\forall \vy \exists \vz \, \Ff( \vu , \vy , \vz )$, we have the conversions: 
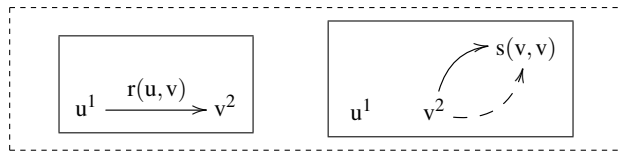
\vs{- \mvs pt}
\[   \forall \vy \exists \vz \, \Ff( \vu , \vy , \vz )  \, \cnvu{\UQRl} \,  
 \cmpl{\BSl \{  \vu , \vv \}   \drs  \{ \frarc{\cmpl{ \exists \vz \, \Ff( \vu , \vy , \vz ) }}{\vu \, \vv } \}   \dps \vu  \ESl}  
 \, \cnvu{\EQRl} \, 
  \cmpl{\BSl \{  \vu , \vv \}   \drs 
  	 \{ \frarc{\cmpl{  \BSl \{ \vu , \vv , \vw  \}  \drs 
	 	 \{\arc{\Ff}{ \vu \, \vv \, \vw} \}   \dps \vu \, \vv  \ESl  }}{\vu \, \vv } \}   \dps \vu  \ESl} \] 
\end{Exmpl}   


	By applying  the $8$    formula rules in any context, one can transform an expression to an equivalent expression without connectives or quantifiers. 
 
 
	 The \emph{equality rule} is the following conversion rule, eliminating expression $\eqs$. 
\Bd 
\item[{\sm ($\eqs$)}]   
$\eqs \, \red \, \BSl \{ \vu \} \drs \ES \dps \Lag \vu, \vu \Rag \ESl$, where $\vu \in \IN$. 
So, we can replace  slice  
$\BSl \{ \vu , \vv, \vw \} \drs
	\{ \arc{\rs}{\vu \, \vv} , \arc{\eqs}{ \vv \, \vw } , \arc{\pss}{\vu \, \vw} \} \dps  \vv \, \vw \ESl$ by the slice 
 $\BSl \{ \vv , \vw \} \drs
 	\{ \arc{\rs}{\vu \, \vv} , \frarc{ \BSl \{ \vu \} \drs \ES \dps \Lag \vu , \vu \Rag \ESl }{ \vv \, \vw } , \arc{\pss}{\vu \, \vw} \} 
 	\dps  \vv \, \vw \ESl$.    
\Ed
 

	By using these $9$ rules, one can eliminate logical symbols and predicates, but  arcs whose expressions are slices or graphs, perhaps complemented, may appear. 
For instance, this happens with  
$\vv \eqs \vw$ and  
$\exists \vy ( \rs( \vu, \vy ) \land \pss( \vy , \vw ))$. 
The following  rules will address these cases. 

	The \emph{complementation  rules} are the following $2$ conversion rules, moving 
$\cmpl{{}^{\hs{\nvs pt}}}$ inside.
\Bd 
\item[{\sm ($\StrFlCUn$)}]   For  an $\rla$-ary graph $\GrH$: 
$\cmpl{\GrH} \, \red \, 
	\BSl \StC{\Lw} \drs \Setof{\arc{\cmpl{\SlT}}{\Lw}}{\SlT \in \GrH} \dps \Lw \ESl$,  
	where $\Lw$ is a list of $\rla$ distinct names.   
So, we can  replace the  complemented $1$-ary graph $\cmpl{\{ \SlS , \SlT  \}}$ by the  slice 
	$\BSl \{ \vv \}  \drs  \{ \arc{\cmpl{\SlS}}{\vv} , \arc{\cmpl{\SlT}}{\vv} \} \dps \vv \ESl$. 
\item[{\sm ($\dc$)}]  $\cmpl{\cmpl{\RE}} \, \red \, \RE$, i.~e. eliminate double complementation. 
\Ed 

	By applying  these  $2$ complementation rules in any context, one can eliminate  arcs whose expressions are complemented graphs.

	The \emph{structural  rules} are the following $3$ conversion rules. 
\Bd 

\item[{\sm ($\StrGrRl$)}]  $\SlS \, \ada \,  \arc{\GrH}{\Lv}  \, \red \, 
	 \Setof{\SlS \, \ada \,  \arc{\SlT}{\Lv}}{ \SlT \in \GrH}$,  i.~e.  
	 replace addition of graph arc by alternative addition of its slice arcs. 
	 So, we replace slice  $\SlS \, \ada \,  \arc{\{ \SlT', \SlT'' \} }{\vu} $ by the graph 
	  $\{ \SlS \, \ada \,  \arc{\SlT' }{\vu}, \SlS \, \ada \,  \arc{\SlT'' }{\vu} \}$. 

\item[{\sm ($\StrSlRl$)}]  $\SlS \, \ada \,  \arc{\SlT}{\Lv}  \, \red \,  \ggl{\SlS} {\Lv} {\SlT}$,  i.~e.  
	replace addition of slice arc by glued slice. 
	
\item[{\sm ($\StrFlLu$)}]  For an $\rla$-ary expression $\RE$: 
$\RE \, \red \,  \BSl \StC{\Lw} \drs \{ \arc{\RE}{\Lw} \}  \dps \Lw  \ESl$, 
where $\Lw$ is a list of $\rla$ distinct names. 
So, for $\rs \in \PS_2$, we can replace $2$-ary expression  
$\rs$ by the $2$-ary slice $\BSl \{  \vu_1 , \vu_2 \}  \drs \{  \arc{\rs}{\vu_1 \, \vu_2} \} \dps \vu_1 \, \vu_2 \ESl$. 

\Ed 

	By  means of     rules ($\StrGrRl$) and  ($\StrSlRl$), one can eliminate  arcs 
whose expressions are graphs or slices.
Rule ($\StrFlLu$) converts expressions to slices and serves to eliminate $\cmpl{\ps}$: 
$\cmpl{\ps} \, \cnvu{\StrFlLu} \, \cmpl{\BSl \StC{\Lu^\rla} \drs \{  \arc{\ps}{\Lu^\rla} \} \dps \Lu^\rla \ESl}$, 
for $\ps \in \PS_\rla$.  


\begin{Exmpl}  \label{Exmpl:DtQntRed}  
Consider the formula  $\rs( \vv , \vw )$. We proceed much as in Example~\ref{Exmpl:DtQntCnv}.
\\  
\smallskip
Formula $\exists \vy \forall \vz \, \rs( \vy , \vz )$ converts to the $0$-ary slice 
$\SlS \seq \BSl \{ \vv \}   \drs  \{ \frarc{  \cmpl{\BSl \{  \vv , \vw \}   \drs  
		\{ \arc{\cmpl{\rs(\vv , \vw)}}{ \vv \, \vw} \}   \dps \vv  \ESl}  }{ \vv } \}   \dps \NL  \ESl$. \\  
This slice $\SlS$ is not basic, but it can be   converted to a basic slice by ($\AtRl$)  as follows: 
\vs{- \nvs pt} 
\[  \BSl \{ \vv \}   \drs  \{ \frarc{  \cmpl{\BSl \{  \vv , \vw \}   \drs  
		\{ \frarc{\cmpl{\rs(\vv , \vw)}}{ \vv \, \vw} \}   \dps \vv  \ESl}  }{ \vv } \}   \dps \NL  \ESl   \, \cnvu{\AtRl} \, 
\BSl \{ \vv \}   \drs  \{ \frarc{  \cmpl{\BSl \{  \vv , \vw \}   \drs  
		\{ \frarc{ \cmpl{ \BSl \{ \vv , \vw  \}  \drs 
		\{ \frarc{\rs}{\vv \, \vw} \} \dps \vv \, \vw \ESl } }{ \vv \, \vw} \}   \dps \vv  \ESl}  }{ \vv } \}   \dps \NL  \ESl \]
\smallskip 
Formula $\forall \vy \exists \vz \, \rs( \vy , \vz )$ converts to the $0$-ary expression 
$\RE \seq \cmpl{\BSl \{  \vv \}   \drs \{ \frarc{\cmpl{  \BSl \{ \vv , \vw  \}  \drs 
	 	 \{\arc{ \rs(\vv , \vw) }{ \vv \, \vw} \}   \dps  \vv  \ESl  }}{ \vv } \}   \dps \NL  \ESl}$. \\ 
This expression $\RE$ is not basic, but it can be   converted to a basic expression $\RF$ by ($\AtRl$) as follows: 
\vs{- \nvs pt}
\[  \cmpl{\BSl \{  \vv \}   \drs  \{ \frarc{\cmpl{  \BSl \{ \vv , \vw  \}  \drs 
	 	 \{\arc{ \rs(\vv , \vw) }{ \vv \, \vw} \}   \dps  \vv  \ESl  }}{ \vv } \}   \dps \NL  \ESl}   \, \cnvu{\AtRl} \, 
\cmpl{\BSl \{  \vv \}   \drs  \{ \frarc{\cmpl{  \BSl \{ \vv , \vw  \}  \drs 
	 	 \{\arc{ \BSl \{ \vv , \vw \}  \drs \{\frarc{\rs}{\vv \, \vw} \} \dps \vv \, \vw \ESl }{ \vv \, \vw} \}   \dps  \vv  \ESl  }}{ \vv } \}   \dps \NL  \ESl} \] 
 Expression $\RF$ can be converted to a basic $0$-ary  slice by ($\StrFlLu$). 
 \end{Exmpl}  


	Rule   ($\StrSlRl$) gives some useful derived rules about arc addition, which we can use to 
	shorten conversions (such shortenings were used in the examples of Section~\ref{sec:Motv}). 
	We can replace addition of: 
a  graph arc  by gluing the graph (($\DrStrGGRl$): $\SlS \, \ada \,  \arc{\GrH}{\Lv}  \, \rtc{\red} \,   \ggl{\SlS}{\Lv}{\GrH}$), 
a  complemented-graph arc  by addition of parallel complemented-slice arcs 
($\SlS \, \ada \,  \arc{\cmpl{\GrH}}{\Lv}  \, \rtc{\red} \, \SlS \, \ada \, \Setof{ \arc{\cmpl{\SlT}}{\Lv}}{ \SlT \in \GrH}$) and 
an equality arc by node renaming ($\SlS \, \ada \, \arc{\eqs}{\vu \, \vv} \, \rtc{\red} \, \SlS \Brn \vu \rpm \vv \Ern$, 
$\SlS \, \ada \, \arc{\eqs}{\vu \, \vv} \, \rtc{\red} \, \SlS \Brn \vv \rpm \vu \Ern$). 

	We can also replace $\rla$ 
conjunctions and disjunctions by slices and graphs, respectively.\fn{For instance, 
with $3$ formulas, we have   
$\rs(\vu, \vv) \land \pss(\vv, \vw)  \land \ps( \vw) \, \rtc{\red} \, 
\BSl \{ \vu , \vv , \vw \}  \drs 
	\{  \arc{\rs(\vu, \vv)}{\vu \, \vv} ,  \arc{\pss(\vv, \vw)}{\vv \, \vw} ,  \arc{\ps(\vw)}{\vw}   \} \dps \vu \, \vv \, \vw \ESl$ and 
$\rs(\vu, \vv) \lor \pss(\vv, \vw)  \lor \ps( \vw) \, \rtc{\red} \, 
\{  \BSl \{ \vu, \vv, \vw \} \drs \{  \arc{\rs( \vu, \vv)}{  \vu \, \vv }  \} \dps  \vu \, \vv \, \vw  \ESl \, , \, 
	\BSl  \{ \vu, \vv, \vw \} \drs \arc{\pss( \vv, \vw)}{ \vv \, \vw }  \} \dps   \vu \, \vv \, \vw  \ESl  \, \,     
	\BSl \{ \vu, \vv, \vw \} \drs \{  \arc{\ps(\vw)}{\vw}  \} \dps  \vu \, \vv \, \vw  \ESl  \}$.).}


\begin{Exmpl}  \label{Exmpl:Unstsf} 
Consider the formula 
$\pss(\vv', \vw') \land \exists \vx [ \rs(\vv', \vx) \land \neg \exists \vy ( \rs(\vx, \vy) \land \pss(\vy, \vw' )) ] $. 
This expression $\RE$ can be converted to the $2$-ary slice 
$\BSl \snt{D} \dps \vv' \, \vw' \ESl$, where  $\snt{D}$ is the draft of Example~\ref{Exmpl:ZrDrft} 
(in~\ref{subsec:Cnstr}). 
\end{Exmpl}  


	We can convert expressions in a modular way. 
	
\begin{Lem}  \label{Lem:ModRed} 
 If $\SlS \, \rtc{\red} \, \GrG$ and $\RE \, \rtc{\red} \, \GrH$, then 
$\SlS \, \ada \,  \arc{\RE}{\Lv} \, \rtc{\red} \, \ggl{\GrG} {\Lv}  {\GrH}$.  
\end{Lem}  
\begin{proof}   By  
($\DrStrGGRl$) rule: 
$\SlS \, \ada \,  \arc{\RE}{\Lv} \, \rtc{\red} \,  \GrG \, \ada \,  \arc{\GrH}{\Lv} \, \seq \, 
\Setof{ \SlP \, \ada \,  \arc{\GrH}{\Lv} }{\SlP \in \GrG} \, \cnvu{\DrStrGGRl} \, 
\Setof{ \ggl{\SlP}{\Lv}{\GrH} }{\SlP \in \GrG} \, = \, \ggl{\GrG} {\Lv}  {\GrH}$. 
\end{proof}  

	Thus, one can obtain a basic form for $\SlS \, \ada \,  \arc{\RE}{\Lv}$ from basic forms 
$\bsf{\SlS}$ and $\bsf{\RE}$, for $\SlS$ and $\RE$.

\begin{Prop} \label{Prop:Red} 
 Every $\rla$-ary expression $\RE$ can be effectively converted to a basic $\rla$-ary  graph $\bsf{\RE}$.
\end{Prop}   
\begin{proof}   By induction on the structure of expressions. 
\end{proof}   


\begin{Exmpl}  \label{Exmpl:Prf}  
Given the predicate symbols of Example~\ref{Exmpl:mor} (in~\ref{subsec:Cnstr}), consider the formula  $\Fp$: 
\vs{ - \nvs pt} 
\[ \qs(\vv , \vw) \, \land \, \exists \vz  \, [ \ps(\vz) \land  \rs(\vv , \vz) \land 
	 \exists \vx \exists \vy \exists \vy' ( \pss(\vv , \vx) \land  \ts(\vx , \vw) \land \as(\vx , \vy) \land \bs(\vy , \vw) )  ]. \]  
 Consider also the formula   
 $\Ft \deq  \exists^3 \vx_1 \vx_2  \vx_3 \,  \exists \vy' \exists^2 \vy_1 \vy_2 \exists \vz' \, \exists^2  \vz_1, \vz_2 \, \Fx$, 
where   $\Fx$ is as follows:  
\vs{ - \nvs pt} 
\[  \ps(\vz') \land \pss(\vv_2 , \vx_3) \land   \ts(\vx_2 , \vz_1) \land 
	\Lm (  \Ba{c} \qs(\vy_1 , \vz_1) \\ \land \\ \qs(\vy_2 , \vz_2) \Ea \Rm )  \land 
	\Lm (  \Ba{c} \as(\vx_1 , \vv')  \\ \land \\ \as(\vx_2 , \vv') \Ea \Rm ) \land 
	\Lm (  \Ba{c}  \rs(\vv , \vz') \\ \land \\ \rs(\vy_1 , \vz') \\ \land \\ \rs(\vy_2 , \vz') \Ea \Rm ) \land 
	 \Lm (  \Ba{c}  \bs(\vy', \vw) \\ \land \\  \bs(\vy', \vz_1) \\ \land \\  \bs(\vy', \vz_2)   \Ea \Rm ).
\] 
Now, form the difference slice 
$\NDS{\{ \FA{\Fp} \} }{\FA{\Ft}} \seq  
\BSl \{ \vv, \vw \}  \drs \{ \arc{\Fp}{ \vv \, \vw },  \arc{\cmpl{\Ft}}{ \vv \, \vw} \dps \NL \} \ESl$. 
Expressions  $\Fp$ and $\Ft$ can be  respectively converted  to the $2$-ary slices  
$\SlS \seq \BSl \snt{D}' \dps \Lag  \vv , \vw \Rag  \ESl$ and   
$\SlT \seq \BSl \snt{D}'' \dps \Lag  \vv , \vw \Rag  \ESl$, 
where $\snt{D}'$ and $\snt{D}''$ are the drafts of  Example~\ref{Exmpl:mor}.  
Thus, we have  
$\NDS{\{ \FA{\Fp} \} }{\FA{\Ft}}   \, \rtc{\red} \,  
\BSl \{ \vv, \vw \}  \drs \{ \arc{\SlS}{ \vv \, \vw },  \arc{\cmpl{\SlT}}{ \vv \, \vw} \dps \NL \} \ESl$. 
Now, we can see that  
$\BSl \{ \vv, \vw \}  \drs \{ \arc{\SlS}{ \vv \, \vw},  \arc{\cmpl{\SlT}}{ \vv \, \vw} \dps \NL \} \ESl 
 \, \cnvu{\StrSlRl} \,   
\BSl \ud{\SlS} \ada  \arc{\cmpl{\SlT}}{ \vv \, \vw} \dps \NL \ESl$.\fn{Indeed: 
$\BSl \{ \vv, \vw \}  \drs \{ \arc{\SlS}{ \vv \, \vw},  \arc{\cmpl{\SlT}}{ \vv \, \vw} \dps \NL \} \ESl \seq  
\BSl \{ \vv, \vw \}  \drs \{ \arc{\cmpl{\SlT}}{ \vv \, \vw} \dps \NL \} \ESl \ada 
\arc{\SlS}{ \vv \, \vw} \, \cnvu{\StrSlRl} \, 
\ggl{ \BSl \{ \vv, \vw \}  \drs \{ \arc{\cmpl{\SlT}}{ \vv \, \vw} \dps \NL \} \ESl } {\Lag \vv , \vw \Rag} {\SlS} \seq  
\BSl \ud{\SlS} \ada  \arc{\cmpl{\SlT}}{ \vv \, \vw} \dps \NL \ESl$.} 
Hence 
$\NDS{\{ \FA{\Fp} \} }{\FA{\Ft}}   \, \rtc{\red} \,   
\BSl \ud{\SlS} \ada  \arc{\cmpl{\SlT}}{ \vv \, \vw} \dps \NL \ESl$. 
\end{Exmpl} 


\subsection{Derivations} \label{subsec:Drv} 
We now introduce the remaining rule and finish the presentation of our   calculus. 

	First, let us review  Examples~\ref{Exmpl:Unstsf} and~\ref{Exmpl:Prf} (in~\ref{subsec:Red}).  
Formula $\RE$  of  Example~\ref{Exmpl:Unstsf} converts  to the $2$-ary slice 
$\BSl \snt{D} \dps \vv' \, \vw' \ESl$, which was  
seen to be zero in Example~\ref{Exmpl:ZrDrft} (in~\ref{subsec:Cnstr}). 
Thus, formula $\RE$ is unsatisfiable. 
Now, consider formulas $\Fp$ and $\Ft$ of  Example~\ref{Exmpl:Prf}, 
where we have seen that 
$\NDS{\{ \FA{\Fp} \} }{\FA{\Ft}}   \, \rtc{\red} \,   
\BSl \ud{\SlS} \ada  \arc{\cmpl{\SlT}}{ \vv \, \vw} \dps \NL \ESl$. 
Now, Example~\ref{Exmpl:mor} (in~\ref{subsec:Cnstr}) shows a morphism 
$\he: \ud{\SlT} \mor \ud{\SlS}$, with 
$\DLST^\he \seq \Lag \vv^\he , \vw^\he \Rag \seq \Lag \vv , \vw \Rag \seq \DLSS$. 
Thus, draft  
$ \ud{\SlS} \ada  \arc{\cmpl{\SlT}}{ \vv \, \vw}$ is zero,  
whence,  slice 
$\NDS{\{ \FA{\Fp} \} }{\FA{\Ft}}$ is null.  
Therefore, we can conclude that $\Fp \cnq \Ft$. 


\begin{Exmpl}  \label{Exmpl:Exprl}	
To introduce expansion and its usefulness, 
consider the  $3$-ary  slice $\SlS$: 
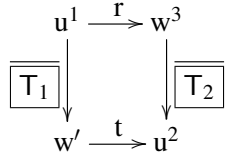
\[  \xymatrix@R30pt@C20pt{
\ar[d]_{{\Ds  \cmpl{\fbox{$
\SlT_1$}
}}}\vu^1\ar[r]^{{\Ds \, \rs}} &\vw^3\ar[d]^{{\Ds  \cmpl{\fbox{$
\SlT_2$}
}}}
\\
\vw'\ar[r]^{{\Ds \, \ts}}&\vu^2
}
\]
where  
$\Ba{ccc} \SlT_1 & := &   \vu^1\, \strl{{\Ds \rs}}{\rar} \,  \vw \, \strl{{\Ds \pss}}{\rar} \, \vv^2 \Ea$, 
$\Ba{ccc} \SlT_2 & := &  \vu^1 \, \strld{\cmpl{\pss}}{\rar} \,  \vw\,  \strld{\ts}{\rar} \,\vv^2  \Ea$.

\noindent 
Slice $\SlS$ is not zero; but in any model $\gM$, the pair $\Bop \g(\vw) , \g(\vw') \Eop$   
is  either in  $\rel{\pss}_{\gM}$ or in  $\rel{\cmpl{\pss}}_{\gM}$. 
So,  $\SlS$ is equivalent to   the $2$-ary graph $\GrG \seq \{ \SlS_+ , \SlS_- \}$, 
with slices $\SlS_+$ and  $\SlS_-$, respectively  as follows: 
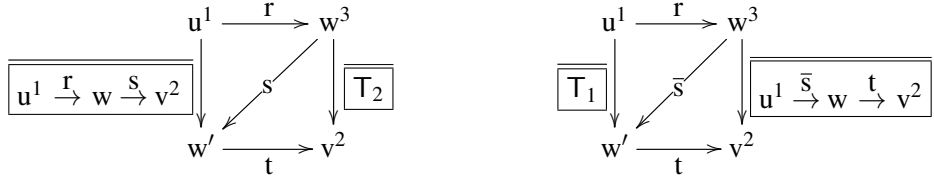
\[ \Ba{ccc} 
 \xymatrix@R32pt@C32pt{
\ar[d]_{{\Ds \cmpl{\fbox{$
\vu^1\, \strl{{\Ds \rs}}{\rar}  \, \vw\strl{{\Ds \pss}}{\rar}   \vv^2  \,  
$}
}
}} \vu^1 \ar[r]^{{\Ds \, \rs}} &\, \vw^3\ar[d]^{{ \, \cmpl{\fbox{$
\SlT_2
$}}
}}\ar[dl]|{{
\Ds \pss
}}
\\
\vw'\ar[r]_{{\Ds \,  \ts }} &\vv^2 
} & \hs{\LVS pt}   \hs{\LVS pt}   & 
\xymatrix@R32pt@C30pt{ 
\ar[d]_{{ \Ds\cmpl{\fbox{$\SlT_1
$}
}
}}\vu^1 \ar[r]^{{\Ds \rs
}}& \, \vw^3\ar[d]^{{\cmpl{\fbox{$\vu^1\strld{\cmpl{\pss}}{\rar} \vw \,\strld{\ts}{\rar}\,\vv^2
$}
}
}}\ar[dl]|{{\Ds \cmpl{\pss}
}}
\\
\vw'\ar[r]_{{\Ds \ts}}&\vv^2
}
\Ea \] 
Slices $\SlS_+$ and  $\SlS_-$ are both zero, so graph $\GrG$ is zero. 
Thus, $\SlS$ is a null slice.  
\end{Exmpl}  


	The expansion rule will 
replace a slice by a graph with $2$ alternative slices. 
\Bd 
\item[{\sm ($\ex$)}]  For an $\ara$-ary  slice $\SlT$ and $\Lv \in {\NS_{\SlS}}^\ara$: 
$\SlS \, \ex \, \{  \ggl{\SlS} {\Lv}  {\SlT} , \SlS \, \ada \,  \arc{\cmpl{\SlT}}{\Lv} \}$.
\Ed 
Note that both $\ggl{\SlS} {\Lv} {\SlT}$ and $\SlS \, \ada \,  \arc{\cmpl{\SlT}}{\Lv}$ are basic whenever 
$\SlS$ and $\SlT$ are basic. 

\begin{Lem} \label{Lem:SndExp} 
For a   slice $\SlS$, an $\ara$-ary slice $\SlT$ and $\Lv \in \IN^\ara$: 
$\SlS \, \eq \, \{  \ggl{\SlS} {\Lv}  {\SlT} , \SlS \, \ada \,  \arc{\cmpl{\SlT}}{\Lv} \}$.
\end{Lem}  
\begin{proof} 
 By 
 Proposition~\ref{Prp:Arcadgl} (in~\ref{subsec:Cnstr}), 
 $\SlS \ada \arc{\SlT}{\Lv} \, \eq \, \ggl{\SlS} {\Lv}  {\SlT}$,  and clearly      
 $\SlS \, \eq \, \{ \SlS \ada \arc{\SlT}{\Lv} , \SlS \, \ada \,  \arc{\cmpl{\SlT}}{\Lv} \}$.
 \end{proof}   

	A \emph{derivation} consists of applications of   the conversion rules and the expansion rule: 
$\Der \, \deq \,  \rtc{( \red \cup \ex)}$. 
A derivation is \emph{normal} iff applications of conversion rules precede applications of the expansion rule: $\RE \,  \rtc{\red} \, \GrG  \,  \rtc{\ex} \, \GrH$.  
In practice, as we wish to derive a zero graph, we may erase slices already found to be zero. 

Let $\Ff$ be the formula 
$\rs(\vu, \vw) \land \ts(\vw', \vv) \land \neg \exists \vx [ \rs(\vu, \vx) \land \pss(\vx, \vv) ] \land 
	\neg \exists \vy [ \neg \pss(\vu, \vy) \land \ts(\vy, \vv) ]$. 
Expression  $\RE \deq \exists \vx \Ff \Brn \vw' \rpm \vx \Ern$ 
converts to the slice $\SlS$ of Example~\ref{Exmpl:Exprl}, where it expands to the  graph $\GrG$.  
We thus have the normal derivation 
$\RE \,  \rtc{\red} \, \SlS  \,  \ex \, \GrG$, with $\GrG$ a zero graph. 
Hence, formula $\Ff$ is unsatisfiable.

\subsection{Soundness and completeness} \label{subsec:SndCmplt}  
We now examine soundness and completeness of our calculus. 

	Soundness is clear (as $\RE \, \eq \, \RF$, whenever $\RE \, \Der \, \RF$): 
if  $\RE \, \Der \, \GrH$ and $\GrH$ is zero, then $\RE$ is null. 
We will show that a converse holds for basic graphs  
(if  $\bsf{\RE}$ is null then $\bsf{\RE}$ expands to a zero graph), and we will have  
completeness of normal derivations: 
if $\RE$ is null, then $\RE \,  \rtc{\red} \, \bsf{\RE}  \,  \rtc{\ex} \, \GrH$, for some zero graph $\GrH$. 

	Henceforth, all sketches, drafts, slices and graphs will be basic. 
We define the following families of slices:  
the family $\ZS$ of  \emph{zero slices} (cf.~\ref{subsec:Cnstr}); 
the family $\EZS$ of \emph{expansivley zero slices}: 
the slices $\SlS$ such that, for some graph $\GrG \incl \ZS$, 
$\SlS \, \rtc{\ex} \, \GrG$;   
the family $\NZS$ of \emph{not expansively zero slices}:   the slices 
outside $\EZS$.


	The following simple properties of these families will be useful.

\begin{Lem} \label{Lem:EZZS} 
For every graph $\GrG$: 
$\GrG \incl \EZS$ iff, for some graph  $\GrH \incl \ZS$, $\GrG \, \rtc{\ex} \, \GrH$.  
\end{Lem}  
\begin{proof} 
 ($\Rar$)  If, for each $\SlS \in \GrG$, $\SlS \, \rtc{\ex} \, \GrH_{\SlS}$ 
 and $\GrH_{\SlS} \incl \ZS$, then, with 
 $\GrH \deq  \bigcup_{\SlS \in \GrG} \, \GrH_{\SlS}$, 
$\GrG \, \rtc{\ex} \, \GrH$ 
and $\GrH  \incl \ZS$.  \\
 ($\Lar$)  if $\GrG \, \rtc{\ex} \, \GrH$, with $\GrH \incl \ZS$, then 
 for each $\SlS \in \GrG$, $\SlS \, \rtc{\ex} \, \GrH_{\SlS}$, with $\GrH_{\SlS} \incl \GrH  \incl \ZS$, 
 whence $\SlS \in \EZS$.  
 \end{proof}   
 
\begin{Lem} \label{Lem:EZEZS} 
For every graph $\GrG$:  
$\GrG \incl \EZS$ iff, for some graph  $\GrH \incl \EZS$, $\GrG \, \rtc{\ex} \, \GrH$.  
\end{Lem}  
\begin{proof} 
 By Lemma~\ref{Lem:EZZS},   since $\ZS \incl \EZS$. 
 ($\Rar$) If $\GrG \incl \EZS$, then $\GrG \, \rtc{\ex} \, \GrH$, with $\GrH \incl  \ZS \incl \EZS$. 
  ($\Lar$) If  $\GrG \, \rtc{\ex} \, \GrH$, with $\GrH \incl \EZS$, then  $\GrH \, \rtc{\ex} \, \GrH'$, 
  with  $\GrH' \incl \ZS$, whence $\GrG \, \rtc{\ex} \, \GrH'$, with $\GrH' \incl \EZS$. 
 \end{proof}   
 
 \begin{Cor} \label{Lem:NEZ} 
For 
$\SlS \in  \NZS$,    $\ara$-slice $\SlT$ and $\Lv \in {\NS_{\SlS}}^\ara$:  
one of $\ggl{\SlS} {\Lv}  {\SlT}$ and $\SlS \, \ada \,  \arc{\cmpl{\SlT}}{\Lv}$ is not expansively zero. 
\end{Cor}  
\begin{proof} 
 By Lemma~\ref{Lem:EZEZS}: if    
 $\{  \ggl{\SlS} {\Lv}  {\SlT} , \SlS \, \ada \,  \arc{\cmpl{\SlT}}{\Lv} \} \incl \EZS$,  then
$\SlS \in  \EZS$.  
 \end{proof}   

	We will show that a slice $\SlS \in  \NZS$ has a model $\gM$ with $\betM{\SlS} \neq \ES$ 

	Given  a slice $\SlS \in  \NZS$, we can obtain a set of slices 
$\SlS_n \seq \BSl \NS_n \drs \AS_n \dps \DLSS_n \ESl$ with $\SlS_n \in \NZS$, for $n \in \Nat$, 
 whose underlying drafts are connected by morphisms 
$\Fm{n}$ from $\ud{\SlS}_n$ to $\ud{\SlS}_{n+1}$, which we extend naturally to morphisms 
$\FM{i} {j}: \ud{\SlS}_i \mor \ud{\SlS}_j$, for $i \le j$. 
Consider  the co-limit of this draft diagram: sketch  $\dg \seq \BDr \NS \drs  \AS \EDr$ 
 with morphisms $\cf{n} : \ud{\SlS}_n \mor \dg$ (cf.~\ref{subsec:Cnstr}). 
 We use this co-limit sketch $\dg$ to define  a \emph{natural model} $\gM$ with   
 $M \deq \NS$, 
 and $\ps^\gM \deq \Setof{ \Lv \in M^\rla}{ \arc{\ps}{\Lv} \in \AS}$,  for $\ps \in \PS_\rla$. 
 
	By construction, the co-limit sketch $\dg$ is saturated  in the following sense: 
given any  $\ara$-ary slice $\SlT \seq \BSl  \ud{\SlT} \dps \DLST \ESl$ and $\Lw \in \NS^\ara$, 
we have arc $\arc{\cmpl{\SlT}}{\Lw} \in \AS$ or 
there is  a morphism $\he : \ud{\SlT} \mor \dg$ with $\DLST^\he = \Lw$.

 	We can establish that satisfying assignments are morphisms.

\begin{Lem} \label{Lem:SatMor} 
Given a 
draft  $\snt{D}$ and $\g : \NS_{\snt{D}} \rar M$, 
 $\g : \snt{D} \rar \gM$ iff  $\g : \snt{D} \mor \dg$. 
\end{Lem}  
\begin{proof} 
 By structural induction (on the total number of complemented slice arcs occurring in $\snt{D}$). 
 \end{proof}   

	Finally, since $\cf{0} : \ud{\SlS} \mor \dg$, we have  $\cf{0}(\DLSS_0) \in \betM{\snt{S}} \neq \ES$. 

	Therefore, if $\GrG \not  \incl \EZS$, then $\GrG$ is not null.

\begin{Thrm} \label{Thrm:SndCmpND} 
Consider an $\rla$-ary expression $\RE$. 
 \Bd 
\item[{\sm ($\Der$)}] If $\RE \, \Der \, \GrH$ and $\GrH$ is zero, then $\RE$ is null.  
\item[{\sm ($\rtc{\red} \cnc  \rtc{\ex}$)}] If $\RE$ is null, then $\RE \,  \rtc{\red} \, \bsf{\RE}  \,  \rtc{\ex} \, \GrH$, for some zero $\rla$-ary graph $\GrH$.  
\Ed  
\end{Thrm}  

\section{Conclusion} \label{sec:Concl} 
We now present some concluding remarks, 
 including comparison with related works. 

	We have presented a refutation    graph calculus for classical first-order predicate logic. 
This sound and complete calculus reduces logical consequence  to establishing that a constructed graph is null, i.~e. has empty extension in every model.  
Our  calculus uses formulas directly and can represent them  by arcs.  

	We have a simple strategy for establishing that a graph $\GrG$ is null:    
first convert  $\GrG$ to basic form, then apply repeatedly the expansion rule, 
erasing slices found to be zero, which is decidable (cf.~\ref{subsec:Cnstr}), 
trying to obtain the empty graph. 
Conversion to basic form, though tedious, can be automated (cf.~\ref{subsec:Red}); 
some ingenuity may be required in selecting which slice of a graph to expand and 
how to do it (cf.~\ref{subsec:Drv}), 
but the embedded slices can provide a finer control.  
In fact, a  (human-guided) system may be envisaged. 

	The idea of using graphical representations for logic appears in several works. 

Girard's proof nets have been applied to classical logic~\cite{Rob_03}, 
where sequent proofs are translated to proof nets.
In our case, however, the (macroscopic) structure of normal derivations is rather simple: 
first conversions, then expansion (cf.~\ref{subsec:Drv}). 

Graph rewriting motivates a graphical representation of first-order predicate logic. 
For the binary fragment 
(with $\eqs$), 
a representation of formulas by graph predicates has been obtained by Rensink~\cite{Ren_04}: 
 a  correspondence between sets  of graph predicates with depth up to $n$ and 
a hierarchy $\exists (\neg \exists)^n$.   
There are close similarities between some concepts  (our sketches are his graphs), 
but his graph predicates involve morphisms (even though they may be reminiscent of our arcs). 

	Our approach does resemble  Peirce's ideas~\cite{Sw_??} as formulated by Dau~\cite{Dau_06}. 
In our representation, we use names only for referring to them in the meta-language: 
if we erase these names, we obtain a representation quite close to the Peirce's ones  
(cf. Example \ref{Exmpl:DtQntRed} in~\ref{subsec:Red}).   
Besides our refutation approach with normal derivations,  there are  some differences: 
we allow formulas directly in the graphs (
and need conversion rules), rather than pre-processing  diagrams for them; 
Peirce considers the fragment $\neg$, $\land$ and $\exists$, 
whereas we use graphs to cope with $\lor$, which seems to lead to less cumbersome  representations; 
we handle $\eqs$, first as a $2$-ary predicate and then as  a special one, 
whereas Peirce represents it directly by identity lines, leading to more compact diagrams. 
So, there appear to be  advantages and disadvantages on both sides. 
 
	Some further work on our calculus would be: 
add function symbols   
(for this purpose, some ideas used for structured nodes~\cite{FVVV_09}  seem promising);  
provide a detailed comparison between it and~\cite{Ren_04} 
(such a comparison between Peirce's  and Rensink's approaches  is reported difficult~\cite{Ren_04}, p. 333);
develop a ``middle-ground'' between our approach and Dau's~\cite{Dau_06}, with the best features from each one.



\nocite{*}
\bibliographystyle{eptcs}
\bibliography{generic}

\end{document}